\documentclass[11pt]{article} 

\usepackage[margin=1in]{geometry}
\usepackage{amsmath,amssymb,amsthm}
\usepackage{bm}
\usepackage{mathrsfs}
\usepackage{graphicx}
\usepackage{booktabs}
\usepackage{microtype}
\usepackage[hidelinks]{hyperref}
\usepackage{natbib}
\usepackage[section]{placeins}
\DeclareMathOperator{\Var}{Var}

\newcommand{\R}{\mathbb{R}}
\newcommand{\cP}{\mathcal{P}}
\newcommand{\E}{\mathbb{E}}
\newcommand{\Law}{\mathcal{L}}

\newcommand{\textcite}{\citet}

\newtheorem{theorem}{Theorem}
\newtheorem{proposition}{Proposition}
\newtheorem{lemma}{Lemma}
\newtheorem{corollary}{Corollary}
\theoremstyle{definition}
\newtheorem{definition}{Definition}
\newtheorem{assumption}{Assumption}
\theoremstyle{remark}
\newtheorem{remark}{Remark}

\title{The Geometry of Heterogeneous Extremes:\\ Optimal Transport and Entropic Design%
\thanks{This paper grew out of many long conversations about how extremes arise in economics and why they matter. I am grateful to several colleagues in economics, and especially to a few in adjacent disciplines, who convinced me that the argument needed more structure to carry conviction. I also gratefully acknowledge \textit{Refine} (\url{https://www.refine.ink/}) for assistance in checking proofs in a preliminary draft. All remaining errors, extreme or otherwise, are mine.}}

\author{I.\ Sebastian Buhai\thanks{Email: \texttt{sebastian.buhai@sofi.su.se}. Full coordinates at \url{https://www.sebastianbuhai.com}.}\\
  \small SOFI at Stockholm University\\
  \small Instituto de Economia at UC Chile\\
  \small NIPE at University of Minho
}

\date{\footnotesize{Version of March 22, 2026. \href{https://www.sebastianbuhai.com/papers/publications/geometry_of_extremes.pdf}{Latest version}.}}

\begin{document}

\maketitle

\begin{abstract}
\footnotesize{
Extreme outcomes depend not only on shock tails but also on heterogeneity in how many opportunities agents get to sample. In the mixed-Poisson search framework, a randomly drawn agent's normalized maximum converges to
\[
H_{\gamma,F}(x)=P_0\!\big(v_\gamma(x)\big),
\]
a Laplace-transform mixture of a classical extreme-value law, with $F$ the mean-one distribution of opportunity intensities.

Treating $F$ as primitive, we study the operator $F\mapsto H_{\gamma,F}$. That operator representation organizes the paper: convex-order comparisons, the homogeneous benchmark, and pointwise cdf bounds are Laplace-analytic consequences. Its main quantitative payoff is geometric. Via the canonical coupling representation $Z_{\gamma,F}=w_\gamma(E/X)$, optimal transport on transformed types yields explicit Wasserstein stability for the entire induced law, integrated control of the corresponding quantile schedule, canonical ambient interpolation paths, and an explicit renormalization bridge back to the mean-one economic slice. We also give a second-order expansion that separates extreme value thery (EVT) approximation error from the heterogeneity kernel.

As a complementary contribution, we study a Kullback-Leibler regularized design problem that chooses $F$ subject to a mean constraint relative to a baseline. For objectives linear in $F$, including expected utility of normalized extremes under canonical representation, the solution is the corresponding exponential tilt, with the heterogeneous-EVT kernel supplying the score. A stylized labor market network application interprets $F$ as the cross-sectional distribution of access to job opportunities, shows how the adapted geometry controls counterfactual movements in the full top wage distribution, and renders explicit that the operational robustness claims are conditional on the maintained metric: the main linear theorem lives in the adapted distance $d_{\gamma,p}$, while the bridge from raw space Wasserstein error can be only H\"older in the economically-relevant Fr\'echet regime. We also distinguish ambient transport geometry from mean-preserving economic counterfactuals and separate finite horizon robustness from asymptotic approximation.

\smallskip \smallskip
\noindent\textbf{JEL codes:} C46; C61; D83; D85; J31; J64.\\
\noindent\textbf{Keywords:} Extreme value theory; heterogeneous search; optimal transport; Wasserstein distance; entropic design; labor market networks.
}
\end{abstract}

\section{Introduction}\label{sec:intro}

\subsection{Motivation and overall contribution}\label{subsec:motivation}
Extreme outcomes at the agent or unit level often shape broader inequality and efficiency patterns.
Examples include a worker's best wage offer, a firm's best sourcing or productivity opportunity, and the best match generated through a network.
In a broad class of models, the realized payoff is the maximum of a set of opportunities.
The size of that set can depend, inter alia, on search, attention, or network position.

Standard extreme value theory (EVT) studies maxima of i.i.d.\ sequences under deterministic sample sizes.
Economic environments often feature heterogeneous opportunity sets: different agents sample different numbers of opportunities, and the draw count is itself random.
A leading economic formulation is the mixed Poisson framework of \citet{Mangin2025}, under which the distribution of a randomly drawn agent's normalized maximum admits a Laplace-transform representation.
If $X \sim F$ indexes an agent's draw intensity and $P_0(z)=\E[\exp(-zX)]$ denotes the Laplace transform of $F$, then the distribution of normalized maxima converges to
\[
H_{\gamma,F}(x)=P_0\big(v_\gamma(x)\big),\qquad v_\gamma(x)=-\log H_\gamma(x),
\]
where $H_\gamma$ is the classical extreme-value limit associated with the underlying opportunity distribution.
This Laplace-mixture form is the analytical backbone of the paper.

Our contribution is to take the heterogeneity distribution itself as the primitive object and to analyze the induced operator $F \mapsto H_{\gamma,F}$ on the space of type distributions.
The novelty claim is deliberately sharp.
Once the Laplace-mixture form is established, convex-order comparisons, benchmark comparisons with the homogeneous economy, and several pointwise statements follow from standard one-dimensional Laplace and Wasserstein arguments.
Those results are useful, but they are not the paper's main quantitative novelty.

The operator representation has a geometric payoff.
Once the canonical coupling representation $Z_{\gamma,F}=w_\gamma(E/X)$ is available, Theorem~\ref{thm:stability_HEV} turns optimal transport on transformed types into explicit Wasserstein control of the entire induced law of extremes.
Corollary~\ref{cor:adapted_geodesic_stability} packages the same geometry into canonical ambient interpolation paths, Proposition~\ref{prop:renormalization_bridge} gives an explicit bridge back to the mean-one economic slice, Corollary~\ref{cor:quantile_schedule_stability} translates the same bound into control of the entire quantile schedule, and Theorem~\ref{thm:second_order_expansion} separates second-order EVT approximation error from the heterogeneity kernel.
These are the results that deliver robustness to misspecification or estimation error in $F$ and quantitative comparative statics for whole laws, rather than only pointwise cdf comparisons.
The main linear statement is formulated in the adapted metric $d_{\gamma,p}$; the bridge back to raw space Wasserstein error can be weaker and is made explicit later.

We also study a complementary normativ problem on the same type space.
There, the object is not a full two-marginal transport problem, and its solution does not rely on the Wasserstein geometry developed for the positive results.
Instead, the planner chooses a new heterogeneity distribution $F$ subject to a mean constraint and a Kullback-Leibler (KL) penalty relative to a baseline $F_0$.
This is a one-marginal entropy projection.
The Gibbs form of the solution is standard in that class of problems; our contribution is to show how the same Laplace kernel that defines $F \mapsto H_{\gamma,F}$ generates the score functions entering the planner's first-order conditions for cdf criteria and for expected utility under the canonical stochastic representation.
Being explicit about that division of labor sharpens both the mathematics and the interpretation.

Throughout, the formal object is the cross-sectional law of a randomly drawn agent's normalized maximum under a common offer distribution.
Aggregate maxima across agents, as well as settings in which heterogeneity also shifts offer quality or induces dependence between opportunity access and offer values, lie outside the present reduced-form scope.

\subsection{Main results}\label{subsec:preview}
The paper contributes four sets of results.

\paragraph{A structural decomposition of heterogeneous extremes.}
We formalize the operator view $T:F \mapsto H_{\gamma,F}$ and separate what is purely (Laplace-)analytic from what is genuinely geometric.
Order comparisons and several pointwise inequalities follow from the convexity of the Laplace kernel.
By contrast, the quantitative comparison of entire laws requires an adapted metric on the type space.
This distinction matters for both interpretation and application.

\paragraph{Wasserstein geometry and quantitative robustness.}
The paper's main new theorem is Theorem~\ref{thm:stability_HEV}: under the regime-specific moment conditions intrinsic to extreme-value theory, the map $F \mapsto H_{\gamma,F}$ is Lipschitz from the ambient transformed type space to the space of extreme laws equipped with Wasserstein distance.
The linear bound is therefore stated in the adapted metric $d_{\gamma,p}$, which is itself an ordinary Wasserstein distance after a monotone transform of types.
The same construction also furnishes canonical constant-speed geodesics in that ambient geometry, summarized in Corollary~\ref{cor:adapted_geodesic_stability}.
When the transform nonlinear, those adapted geodesics need not preserve the economic normalization $\E[X]=1$ at intermediate times; Proposition~\ref{prop:renormalization_bridge} therefore provides an explicit quantitative bridge from the ambient path back to the mean-one economic slice.
The same whole law control also yields integrated bounds for the full quantile schedule of normalized extremes, not only fixed threshold cdf comparisons.
Operationally, this means that one either controls $d_{\gamma,p}$ directly or uses a separate bridge from raw space error, which can be only H\"older in Fr\'echet settings, unless extra support restrictions make the transform Lipschitz.\footnote{In plain terms, the theorem controls not just one tail probability but the whole distribution of normalized extremes, and hence the whole quantile curve, in an integrated sense. In applications, the most direct route is to measure heterogeneity in the adapted metric $d_{\gamma,p}$. If one instead starts from error in the original type space, an additional comparison is needed. In Fr\'echet settings that comparison may be weaker than linear unless support restrictions keep the relevant transform Lipschitz, for instance by keeping the support away from zero.}

\paragraph{A complementary entropy-regularized design problem on the same type space.}
For objectives that are linear in $F$, including expected utility of normalized extremes under the canonical stochastic representation, we solve the planner's problem with a KL penalty and a mean constraint.
The optimizer is the corresponding exponential tilt of the baseline.
This problem is best understood as a one-marginal entropy projection.
Its tight connection to the positive analysis is kernel-based rather than transport-based: the score functions entering the tilt are generated by the same heterogeneous-EVT primitives that determine the operator $F \mapsto H_{\gamma,F}$ and its directional derivative.

\paragraph{Application to labor market networks.}
We apply the framework to a reduced-form labor market environment in which network position governs offer arrival while the offer distribution remains common across workers.
The opportunity distribution induced by degree, weighted degree, or other validated proxies for access to search translates into a heterogeneous extreme-value law for a randomly drawn worker's normalized top wage.
The positive results quantify how network inequality changes that cross-sectional right tail distribution, give a concrete whole law counterfactual for the full schedule of top wage quantiles, and make explicit the conditions under which measurement error in network heterogeneity propagates into tail prediction error.

\subsection{Empirical interface and identification of heterogeneity}\label{subsec:empirical_interface}
The empirical content of the framework is promising, thogh one should be careful about what is and is not identified.
If a cross section of offer counts is observed at a fixed horizon, then the data identify the count distribution and its pgf.
Under the mixed-Poisson structure, that pgf identified from data pins down the Laplace transform $P_0$ on the interval $z\in[0,\theta]$.
Because $P_0$ is a Laplace transform, knowledge of that transform on any open interval identifies the underlying mixing distribution $F$ at the population level.
The real difficulty is statistical rather than logical: analytic continuation and Laplace inversion are severely ill-posed, so richer variation, e.g., multiple horizons, repeated exposures, or parametric/semi-parametric structure, remains highly valuable for stable estimation rather than for point identification per se.

A second empirical route uses direct proxies for opportunity intensity, such as normalized degree, referral exposure, centrality, or estimated arrival rates.
That route identifies an empirical approximation to $F$ after one specifies a modeling map from the observed proxy into the latent intensity entering the mixed-Poisson law.
Likewise, network data by themselves do not identify the adapted metric $d_{\gamma,p}$; that metric becomes operational once the researcher has committed to a tail index $\gamma$ and to a maintained measurement model linking the observed proxy to the latent type entering the relevant stability theorem.

The same caution applies to inference from extremes.
Using the heterogeneous extreme-value limit to infer $F$ from tail behavior requires the tail-limit ingredients, i.e. the index $\gamma$, the normalizations $(a_\theta,b_\theta)$, and, for second-order refinements, the corresponding second-order objects, either to be known or to be estimated in a first step.
The full parent law $G$ need not be known for the EVT approximation itself.
Accordingly, our paper emphasizes robustness and error propagation more than stable nonparametric recovery.
The main operational message is therefore conditional: once the analyst has an estimate $\widehat F$ and justified control of the adapted metric $d_{\gamma,p}$ around $\widehat F$, the stability results convert uncertainty about heterogeneity into explicit uncertainty about tail predictions.
If one starts instead from raw space Wasserstein error, Proposition~\ref{prop:type_metric_controls} supplies the bridge.
In particular, when $0<\gamma<1$ and the assumptions of Theorem~\ref{thm:stability_HEV} hold with $p\gamma<1$, that bridge is only H\"older of order $\gamma$, so Fr\'echet-side error propagation can be materially slower than linear.

\subsection{Related literature}\label{subsec:literature}
\paragraph{Heterogeneous extreme value theory.}
We build on \citet{Mangin2025}; an earlier circulation is \citet{BeckerMangin2023}.
On the probability side, the mixed-Poisson law $H_{\gamma,F}$ sits in the classical lineage on maxima under random sample size or random indexing \citep{BarndorffNielsen1964,Galambos1973,SilvestrovTeugels1998}, with related max-geometric and max-semistable notions in, e.g., \citet{RachevResnick1991} and \citet{Megyesi2002}.
Our departure from that literature is to treat the heterogeneity distribution itself as the primitive object and to study the induced operator $F\mapsto H_{\gamma,F}$ through order, metric, and variational tools.
Adjacent econometric work by \citet{EinmahlHe2023} studies extreme-value estimation under heterogeneous marginals rather than heterogeneity in the intensity of opportunity arrival.
For classical EVT background, see \citet{deHaanFerreira2006} and \citet{Resnick2008}; for complementary peaks-over-threshold foundations, see \citet{BalkemaDeHaan1974} and \citet{Pickands1975}.

\paragraph{Optimal transport, coupling, and entropy projection.}
Our positive results use optimal transport as a geometry on distributions \citep{Villani2009,Galichon2016} and are close in spirit to the coupling approach to extremes in \citet{BobbiaDombryVarron2019}; complementary Wasserstein-type control of Fr\'echet approximation via Stein methods is developed in \citet{MansanarezPolySwan2025}.
On the normative side, our planner problem is a one-marginal entropy projection with a linear moment constraint.
The duality and exponential-tilt characterization connect to the Gibbs variational principle \citep{DonskerVaradhan1975} and, more loosely, to modern Schr\"odinger and entropic optimal transport treatments \citep{Cuturi2013,Leonard2014,PeyreCuturi2019,GhosalNutzBernton2022,Nutz2022}.
We therefore use optimal transport and entropy regularization in related but distinct ways in the positive and normative parts of the paper.

\paragraph{Labor market networks.}
Our application draws on the literature on job information and referrals in networks, including \citet{Granovetter1973}, \citet{Montgomery1991}, \citet{Topa2001}, \citet{CalvoArmengolJackson2004}, and \citet{IoannidesLoury2004}, as well as more recent work on referral inequality and segregation such as \citet{BuhaiVanderLeij2023} and \citet{BolteImmorlicaJackson2020}.
Relative to that literature, we emphasize a tail-sensitive reduced form: how the distribution of opportunities maps into the distribution of worker-level extremes, how robust that mapping is to perturbations in heterogeneity, and how policy might reshape the opportunity distribution when the right tail is economically consequential.

\subsection{Roadmap and regime map}\label{subsec:roadmap}
\noindent\textbf{Structure of the paper.}
Section~\ref{sec:env} introduces the heterogeneous extreme-value environment.
Section~\ref{sec:geometry} develops order properties, misallocation indices, and geodesic structure on the type space.
Section~\ref{sec:Wstab} contains the main metric stability results, finite horizon robustness bounds, and the second-order expansion.
Section~\ref{sec:entropic} studies the complementary entropy-regularized design problem.
Section~\ref{sec:apps} presents the labor market network application.
Section~\ref{sec:concl} concludes.
Appendix~A contains proofs and Appendix~B records technical tools.

Before moving on to the main sections, we already include Table~\ref{tab:regime_map} here in the Introduction, as a pragmatic overview of the paper's regime-specific scope.
A recurring theme is that lower-support restrictions are needed only when logarithmic, negative-power, or inverse-power transforms of $X$ enter the argument.
Another is that the adapted geodesics of Section~\ref{sec:Wstab} live in an ambient transformed type space; when $s_\gamma$ is nonlinear, they need not preserve the mean-one normalization at intermediate times.
Proposition~\ref{prop:renormalization_bridge} quantifies the discrepancy between the ambient path and its mean-one renormalization, while the raw space geodesics and pointwise comparisons in Section~\ref{sec:geometry} remain available when exact mean preservation is required throughout.
\begin{table}[ht]
\centering
\footnotesize
\caption{Scope of the main analytical results across extreme-value regimes}
\label{tab:regime_map}
\begin{tabular}{p{0.34\textwidth}p{0.17\textwidth}p{0.39\textwidth}}
\toprule
Result & Regime & Additional conditions \\
\midrule
Proposition~\ref{prop:HEV}, Proposition~\ref{prop:laplace_order}, Corollary~\ref{cor:pointwise_cdf_stability}, Proposition~\ref{prop:finite_theta_stability}, Theorem~\ref{thm:second_order_expansion} & all $\gamma$ & unified baseline Assumptions~\ref{ass:doa}--\ref{ass:mixedPoisson}; no positive lower support bound is imposed \\
Proposition~\ref{prop:pushforward_rep} and Lemma~\ref{lem:kernel_expectation_frechet} & Fr\'echet / $\gamma>0$ & product representation specific to the Fr\'echet case; no support bound away from zero is required there \\
Theorem~\ref{thm:stability_HEV} & all $\gamma$ & $d_{\gamma,p}(F_1,F_2)<\infty$ is the primitive domain condition; equivalently, one needs finite transformed $p$th moments, namely $X^{\gamma p}$ for $\gamma>0$, $|\log X|^p$ for $\gamma=0$, and $X^{-|\gamma|p}$ for $\gamma<0$ \\
Proposition~\ref{prop:type_metric_controls} & $\gamma\in(0,1]$ or $\gamma\le 0$ & for $\gamma\in(0,1]$, no lower support bound is needed; for $\gamma\le 0$, support on $[a,\infty)$ is imposed explicitly \\
Section~\ref{subsec:closed_form} power and inverse-power tilts & all regimes in principle & inverse-power scores require baseline support bounded away from zero; positive-power scores require the baseline tail to be light enough for the exponential tilt to be integrable \\
\bottomrule
\end{tabular}
\end{table}

\section{Environment: heterogeneous extreme value primitives}\label{sec:env}

\subsection{Baseline extreme value theory}\label{subsec:evt_setup}
Let $\{Y_j\}_{j\ge 1}$ be i.i.d.\ real valued random variables with common cdf $G$.
Let $M_n := \max_{1\le j\le n} Y_j$ denote the sample maximum.

We impose a standard maximum domain of attraction condition.
Equivalently, after affine normalization, the law of $M_n$ converges to a generalized extreme value distribution.

\begin{assumption}[Domain of attraction]\label{ass:doa}
There exist sequences $a_n>0$ and $b_n\in\R$ and a parameter $\gamma\in\R$ such that
\[
\Pr\!\left(\frac{M_n-b_n}{a_n}\le x\right)\to H_\gamma(x)
\quad\text{as }n\to\infty,
\]
for all continuity points $x$ of $H_\gamma$.
\end{assumption}

We use the canonical generalized extreme value family
\[
H_\gamma(x)=
\begin{cases}
\exp\!\left(-(1+\gamma x)^{-1/\gamma}\right), & \gamma\neq 0,\ 1+\gamma x>0,\\
\exp\!\left(-e^{-x}\right), & \gamma=0,
\end{cases}
\]
with the conventions $H_\gamma(x)=0$ for $\gamma>0$ and $x\le -1/\gamma$, and $H_\gamma(x)=1$ for $\gamma<0$ and $x\ge -1/\gamma$.
Define the tail transform
\[
v_\gamma(x) := -\log H_\gamma(x),
\]
so that $v_\gamma(x)=(1+\gamma x)^{-1/\gamma}$ for $\gamma\neq 0$ on $\{1+\gamma x>0\}$ and $v_0(x)=e^{-x}$ on $\R$.

\subsection{Mixed Poisson heterogeneity in draw counts}\label{subsec:mixed_poisson}
Each agent draws a random number $N(\theta)$ of i.i.d.\ opportunities, where $\theta>0$ scales the average arrival rate.
Following \citet{Mangin2025}, the draw count is mixed Poisson.

\begin{assumption}[Mixed Poisson heterogeneity]\label{ass:mixedPoisson}
Let $X\sim F$ on $[0,\infty)$ with $\E[X]=1$ and $\Pr(X=0)=0$.
Conditional on $X$, the draw count is Poisson:
\[
N(\theta)\mid X \sim \mathrm{Poisson}(\theta X),\qquad \theta>0.
\]
The sequence $\{Y_j\}_{j\ge 1}$ is independent of $(X,N(\theta))$.
\end{assumption}

\begin{remark}[Why we exclude an atom at zero]\label{rem:zero_mass}
The restriction $\Pr(X=0)=0$ is mainly a unified domain convention rather than a substantive claim that near-zero opportunity types are unimportant.
Mass arbitrarily close to zero is allowed throughout the baseline model.
The reason to exclude an atom exactly at zero is that, for $\gamma\le 0$, it produces an atom at $-\infty$ in the normalized limit law, while several metric constructions below use $\log x$ or $x^\gamma$ and are therefore naturally formulated on $(0,\infty)$.
When a positive lower support bound is genuinely needed, we impose it explicitly in the corresponding proposition.
\end{remark}

The distribution of $N(\theta)$ is mixed Poisson with mean $\E[N(\theta)]=\theta$.
Let
\[
P_0(z):=\E\!\left[e^{-zX}\right],\qquad z\ge 0,
\]
denote the Laplace transform of $F$.
The next lemma summarizes the key generating function identity that will be used repeatedly.

\begin{lemma}[Probability generating function of mixed Poisson]\label{lem:mp_pgf}
Under Assumption~\ref{ass:mixedPoisson}, for any $y\in[0,1]$ and any $\theta>0$,
\[
\E\!\left[y^{N(\theta)}\right]=P_0\!\left(\theta(1-y)\right).
\]
\end{lemma}

\begin{remark}[Search technology versus type heterogeneity]
Assumption~\ref{ass:mixedPoisson} admits two complementary interpretations.
One may take the mixed Poisson law of $N(\theta)$ as the primitive ``search technology'', in which case $F$ is an equivalent representation.
Alternatively, one may take $F$ as the primitive distribution of types $X$, where higher $X$ implies a higher mean draw count $\theta X$.
The Laplace transform $P_0$ summarizes the search technology and is the only object from $F$ that enters the distribution of extremes.
The mixed Poisson specification is analytically valuable precisely because the conditional \emph{distribution function} of the maximum takes the exponential form
\[
\Pr(M_\theta\le x\mid X)=\exp\{-\theta X(1-G(x))\},
\]
so heterogeneity enters through a Laplace transform rather than through a more complicated count generating function.
Its complement gives the corresponding conditional tail probability.
This special role of Poisson search is also emphasized in \citet{BeckerMangin2023}, which contrasts the Poisson benchmark with more general count technologies.
\end{remark}

\subsection{Heterogeneous maxima and their distribution}\label{subsec:hev}
Given $N(\theta)$ draws, the realized outcome is the maximum
\[
M_\theta := \sup_{1\le j\le N(\theta)} Y_j,
\]
with the convention $\sup\emptyset=-\infty$.

For any $x\in\R$,
\[
\Pr(M_\theta\le x \mid X)=\E\!\left[G(x)^{N(\theta)}\mid X\right]
=\exp\!\left(-\theta X(1-G(x))\right),
\]
and hence, by using Lemma~\ref{lem:mp_pgf},
\begin{equation}\label{eq:cdf_Mtheta}
\Pr(M_\theta\le x)=P_0\!\left(\theta(1-G(x))\right).
\end{equation}

\subsection{Heterogeneous extreme value limit law}\label{subsec:hev_limit}
We now combine the domain of attraction condition with the mixed Poisson structure.
Define $a_\theta:=a_{\lfloor \theta\rfloor}$ and $b_\theta:=b_{\lfloor \theta\rfloor}$ for $\theta\ge 1$.
This step extension is convenient for the first-order limit theorem below.
In the second-order subsection later on, we revert to continuous normalizing functions, as is standard in extreme value theory, to avoid discretization artifacts.

\begin{proposition}[Heterogeneous extreme value law]\label{prop:HEV}
Suppose Assumptions~\ref{ass:doa} and~\ref{ass:mixedPoisson} hold.
Then for all continuity points $x$ of $H_{\gamma,F}$,
\[
\Pr\!\left(\frac{M_\theta-b_\theta}{a_\theta}\le x\right)\to H_{\gamma,F}(x)
:=P_0\!\big(v_\gamma(x)\big)
\quad \text{as }\theta\to\infty.
\]
Moreover, if $F=\delta_1$, then $P_0(z)=e^{-z}$ and $H_{\gamma,F}=H_\gamma$.
\end{proposition}

\begin{remark}[Formal scope of the limit law]\label{rem:formal_scope_limit}
The object $H_{\gamma,F}$ is the limit distribution of the normalized maximum for a randomly drawn agent under a common offer distribution $G$.
It is therefore a cross-sectional agent-level object.
The paper does not analyze the economy-wide maximum $\max_i M_{i,\theta}$, nor settings in which network position also changes the distribution of the ofer quality or induces dependence between $X$ and the $Y_{ij}$.
Those extensions are potentially important, but they require additional structure beyond the present mixed-Poisson layer.
\end{remark}

\begin{remark}[Connection to random-sample-size maxima]\label{rem:random_sample_size}
Proposition~\ref{prop:HEV} is a mixed-Poisson instance of the classical literature on maxima with random sample size or random indexing; see, e.g., \citet{BarndorffNielsen1964}, \citet{Galambos1973}, and \citet{SilvestrovTeugels1998}.
In that literature, random indexing changes the form of the limit law while often preserving the underlying extreme-value type.
The value added here is not another domain-of-attraction result, but the decision to treat the mean one heterogeneity distribution $F$ and its Laplace transform $P_0$ as primitives for the order, metric, and design analysis developed below.
\end{remark}

\paragraph{Comment.}
The representation \eqref{eq:cdf_Mtheta} reduces the asymptotics of $M_\theta$ to the behavior of $\theta(1-G(b_\theta+a_\theta x))$.
Assumption~\ref{ass:doa} implies that $G(b_\theta+a_\theta x)^{\lfloor\theta\rfloor}\to H_\gamma(x)=e^{-v_\gamma(x)}$ and thus $\theta(1-G(b_\theta+a_\theta x))\to v_\gamma(x)$ at continuity points.
The formal argument is given in Appendix~A.
Our focus on block maxima is natural in search environments where the economic object is the best offer over a horizon.
Complementary peaks-over-threshold formulations are classical in extreme value theory (see, e.g., \citet{Pickands1975} and \citet{BalkemaDeHaan1974}).

\subsection{A canonical product representation in the Fr\'echet case}\label{subsec:pareto_rep}
A particularly useful representation emerges under Fr\'echet normalization, which is natural for heavy-tailed primitives.
Suppose $\gamma>0$ and work with the standard Fr\'echet cdf
\[
H_\gamma^{\mathrm{Fr}}(z)=\exp\!\left(-z^{-1/\gamma}\right),\qquad z>0.
\]
If $Z_\gamma$ is a Fr\'echet random variable with cdf $H_\gamma^{\mathrm{Fr}}$ and $Z_\gamma$ is independent of $X$, then
\[
\Pr\!\left(X^\gamma Z_\gamma \le z\right)=\E\!\left[\exp\!\left(-X z^{-1/\gamma}\right)\right]=P_0\!\left(z^{-1/\gamma}\right),\qquad z>0.
\]
Thus the heterogeneous Fr\'echet limit admits the product representation
\begin{equation}\label{eq:product_rep}
Z \ \overset{d}{=} \ X^\gamma Z_\gamma,
\end{equation}
where $Z$ has cdf $z\mapsto P_0(z^{-1/\gamma})$ on $(0,\infty)$.
Equivalently, the canonically normalized GEV variable
\[
Z_{\mathrm{GEV}}:=\frac{Z-1}{\gamma}
\]
has cdf $H_{\gamma,F}$.
We shall use \eqref{eq:product_rep} to construct couplings and to obtain sharp Wasserstein bounds in Section~\ref{sec:Wstab}.

\subsection{Extreme outcome functionals}\label{subsec:outcome_functionals}
The paper studies tail-sensitive functionals of the distribution of extremes.
Let $\zeta:\R\to\R$ be a measurable payoff function.
For finite $\theta$, define
\[
\Phi_\theta(F):=\E[\zeta(M_\theta)].
\]
For normalized extremes, define
\[
Z_\theta := \frac{M_\theta-b_\theta}{a_\theta}, \qquad
\Psi_\theta(F):=\E[\psi(Z_\theta)]
\]
for measurable $\psi$ such that the expectation is well defined.
Section~\ref{sec:Wstab} establishes pointwise finite horizon stability for the cdf of $Z_\theta$, Wasserstein stability for the limit law $F\mapsto \Law(Z)$, and derives Lipschitz-type bounds for classes of functionals $\psi$ that are natural for tail analysis.

Section~\ref{sec:entropic} uses these objects to define and solve an entropy-regularized design problem for $F$.

\section{Geometry of heterogeneous extreme value laws}\label{sec:geometry}

This section develops structural and geometric properties of the heterogeneous extreme value operator introduced in Section~\ref{sec:env}.
Throughout, $F$ denotes a probability measure on $[0,\infty)$ with $\E_F[X]=1$ and $\Pr_F(X=0)=0$,
and we write
\[
P_0(z):=\E_F\!\left[e^{-zX}\right]=\int_{0}^{\infty} e^{-zx}\,F(dx),\qquad z\ge 0,
\]
for the Laplace transform of $F$.
For a fixed tail index $\gamma\in\R$ and tail transform $v_\gamma(x):=-\log H_\gamma(x)$
(defined on $\{x:H_\gamma(x)\in(0,1)\}$), the heterogeneous extreme value law is
\[
H_{\gamma,F}(x)=P_0\!\big(v_\gamma(x)\big).
\]
We view the mapping
\[
T:\ F \mapsto H_{\gamma,F}
\]
as an operator from a space of type distributions to a space of extreme value laws.
The section has two roles.
Subsection~\ref{subsec:order} records order and benchmark consequences of the Laplace-mixture structure.
Those results are useful for interpretation and for the application, but they are not the core geometric inovation.
The later subsections then isolate the coupling and variational ingredients needed for the genuinely geometric stability results in Section~\ref{sec:Wstab} and, more indirectly, for the complementary design problem in Section~\ref{sec:entropic}.

\subsection{Order comparisons and extreme misallocation}\label{subsec:order}

A central economic question is how dispersion in opportunities changes extreme outcomes.
In our setting, dispersion enters through the distribution $F$ of draw intensities $X$.
A natural formalization is convex order, which coincides with mean-preserving spreads.

\begin{definition}[Convex order]\label{def:convex_order}
Let $F_1,F_2$ be probability measures on $[0,\infty)$ with finite first moment and common mean.
We write $F_2 \succeq_{\mathrm{cx}} F_1$ if
\[
\int \varphi(x)\,F_2(dx)\ \ge\ \int \varphi(x)\,F_1(dx)
\]
for every convex function $\varphi:[0,\infty)\to\R$ for which both expectations are finite.
When $F_2 \succeq_{\mathrm{cx}} F_1$, $F_2$ is a mean preserving spread of $F_1$ in the sense of \citet{RothschildStiglitz1970}.
\end{definition}

Convex order has an immediate implication for Laplace transforms because $x\mapsto e^{-zx}$ is convex on $[0,\infty)$ for every $z\ge 0$.

\begin{proposition}[Convex order implies Laplace order]\label{prop:laplace_order}
Let $F_1,F_2$ be probability measures on $[0,\infty)$ with $\E[X]=1$.
If $F_2 \succeq_{\mathrm{cx}} F_1$, then for every $z\ge 0$,
\[
P_0^{(2)}(z)\ \ge\ P_0^{(1)}(z).
\]
Consequently, for every $x$ such that $v_\gamma(x)<\infty$,
\[
H_{\gamma,F_2}(x)\ \ge\ H_{\gamma,F_1}(x).
\]
Equivalently, if $Z_i$ has cdf $H_{\gamma,F_i}$, then $\Pr(Z_2>t)\le \Pr(Z_1>t)$ for all $t\in\R$.
\end{proposition}

\begin{corollary}[Heterogeneity lowers extremes relative to the homogeneous benchmark]\label{cor:hetero_vs_hom}
Let $F$ be a probability measure on $[0,\infty)$ with $\E[X]=1$.
Then for every $z\ge 0$,
\[
P_0(z)\ \ge\ e^{-z},
\]
with strict inequality for all $z>0$ whenever $F\neq \delta_1$.
Consequently, for every $x$ such that $v_\gamma(x)<\infty$,
\[
H_{\gamma,F}(x)\ \ge\ H_\gamma(x),
\]
so the heterogeneous extreme value limit is stochastically smaller than the homogeneous limit.
\end{corollary}

\begin{remark}[Tail index versus level effects]\label{rem:first_order_tail_equivalence}
Order comparisons such as Proposition~\ref{prop:laplace_order} operate at the level of the entire distribution function.
They are consistent with the tail index invariance emphasized by \citet{Mangin2025}: heterogeneity in the mixed Poisson draw counts changes the level of extreme outcomes but does not change the limiting tail index $\gamma$.
Under the mean-one normalization, one can actually say more: since $0\le (1-e^{-zX})/z \le X$ for $z>0$ and $\E[X]=1$, dominated convergence gives $P_0(z)=1-z+o(z)$ as $z\downarrow 0$.
Consequently, $1-H_{\gamma,F}(x)=1-P_0\!\big(v_\gamma(x)\big)\sim v_\gamma(x)\sim 1-H_\gamma(x)$ whenever $v_\gamma(x)\downarrow 0$.
Thus heterogeneity changes distributional levels and high quantiles, but it does not change the first-order ultra-tail mass.
\end{remark}

The previous results motivate quantitative measures of how far $F$ is from the egalitarian benchmark $\delta_1$.

\begin{definition}[Extreme misallocation indices]\label{def:misallocation}
Fix the benchmark $F^{\mathrm{eq}}:=\delta_1$.
\begin{enumerate}
\item For $p\ge 1$ and $F\in\mathcal{P}_p([0,\infty))$, define the metric misallocation index
\[
\mathcal{M}_p(F):=W_p(F,\delta_1).
\]
\item Further define the Laplace misallocation curve
\[
\Delta_F(z):=P_0(z)-e^{-z},\qquad z\ge 0.
\]
\end{enumerate}
\end{definition}

\begin{remark}[Interpretation]
$\mathcal{M}_p(F)$ measures opportunity dispersion in a transportation geometry and will be directly compatible with the stability bounds in Section~\ref{sec:Wstab}.
Because the benchmark is a point mass, the first two indices are especially transparent: $\mathcal{M}_1(F)=\E[|X-1|]$ and, under $\E[X]=1$, $\mathcal{M}_2(F)=\sqrt{\E[(X-1)^2]}=\sqrt{\Var(X)}$.
The Laplace curve $z\mapsto \Delta_F(z)$ measures the induced distortion in extreme value laws at the Laplace scale.
By Corollary~\ref{cor:hetero_vs_hom}, $\Delta_F(z)\ge 0$ for all $z\ge 0$, and $\Delta_F\equiv 0$ if and only if $F=\delta_1$.
In particular, these indices quantify how far the opportunity distribution is from the homogeneous benchmark before Section~\ref{sec:Wstab} translates that gap into explicit bounds on the induced law of extremes.
\end{remark}

\subsection{Wasserstein geometry on the type space}\label{subsec:wass_type}

Now let $\mathcal{P}_p([0,\infty))$ denote the set of probability measures on $[0,\infty)$ with finite $p$th moment.
For $\mu,\nu\in\mathcal{P}_p([0,\infty))$, the Wasserstein $p$ distance is
\[
W_p(\mu,\nu):=\left(\inf_{\pi\in\Gamma(\mu,\nu)} \int_{[0,\infty)^2} |x-y|^p\,\pi(dx,dy)\right)^{1/p},
\]
where $\Gamma(\mu,\nu)$ is the set of couplings with marginals $\mu$ and $\nu$.
We refer to \citet{Villani2009} and \citet{Galichon2016} for background.

In one dimension, the geometry is particularly explicit.

\begin{lemma}[Quantile representation and canonical geodesics in one dimension]\label{lem:quantile_geodesic}
Let $\mu,\nu\in\mathcal{P}_p([0,\infty))$ and let $Q_\mu,Q_\nu$ denote their quantile functions.
Then
\[
W_p^p(\mu,\nu)=\int_0^1 |Q_\mu(u)-Q_\nu(u)|^p\,du.
\]
Moreover, the path $(\mu_t)_{t\in[0,1]}$ defined by
\[
Q_{\mu_t}(u)=(1-t)\,Q_\mu(u)+t\,Q_\nu(u),\qquad u\in(0,1),
\]
is a constant-speed $W_p$ geodesic from $\mu$ to $\nu$.
For $p>1$ it is the unique geodesic, while for $p=1$ it is the canonical monotone geodesic and uniqueness need not hold.
\end{lemma}

This one-dimensional characterization has a direct implication for convexity of expectations along the canonical monotone geodesic.

\begin{proposition}[Convexity of expectations along type geodesics]\label{prop:convexity_along_geodesic}
Let $\mu,\nu\in\mathcal{P}_1([0,\infty))$ and let $(\mu_t)_{t\in[0,1]}$ be the canonical monotone Wasserstein geodesic from $\mu$ to $\nu$.
Let $\varphi:[0,\infty)\to\R$ be convex and such that $\int |\varphi|\,d\mu_t<\infty$ for all $t\in[0,1]$.
Then $t\mapsto \int \varphi(x)\,\mu_t(dx)$ is convex on $[0,1]$.
\end{proposition}

\subsection{The extreme value map as a pushforward and coupling transform}\label{subsec:pushforward}

The map $T$ admits a useful representation in cases where the limit law can be written as a measurable image of $(X,\Xi)$ for an auxiliary variable $\Xi$ independent of $X$.
This representation is the basis for the coupling constructions in Section~\ref{sec:Wstab} and for the variational arguments in Section~\ref{sec:entropic}.

\begin{proposition}[Pushforward representations in Fr\'echet and Gumbel cases]\label{prop:pushforward_rep}
\begin{enumerate}
\item \textit{Fr\'echet case} ($\gamma>0$). Let $Z_\gamma$ have the standard Fr\'echet law on $(0,\infty)$. Its cdf is $z\mapsto \exp(-z^{-1/\gamma})$, and $Z_\gamma$ is independent of $X\sim F$.
Then the heterogeneous Fr\'echet limit $Z$ satisfies
\[
Z\ \overset{d}{=}\ X^\gamma Z_\gamma,
\]
and its cdf is $z\mapsto P_0(z^{-1/\gamma})$ on $(0,\infty)$.
Equivalently, the canonically normalized variable
\[
Z_{\mathrm{GEV}}:=\frac{Z-1}{\gamma}
\]
is distributed according to $H_{\gamma,F}$.

\item \textit{Gumbel case} ($\gamma=0$). Let $E$ be exponential with mean one, independent of $X\sim F$.
Define $Z:=\log(X/E)$.
Then $Z$ has cdf $x\mapsto P_0(e^{-x})$, hence $Z$ is distributed according to $H_{0,F}$.
\end{enumerate}
\end{proposition}

A key implication is that many extreme outcome functionals are linear in the mixing distribution $F$ once the auxiliary variable is integrated out.

\begin{lemma}[Kernel representation of expectations in the Fr\'echet case]\label{lem:kernel_expectation_frechet}
Assume $\gamma>0$ and let $Z=X^\gamma Z_\gamma$ be as in Proposition~\ref{prop:pushforward_rep}.
Let $\psi:(0,\infty)\to\R$ be measurable and assume $\E[|\psi(Z)|]<\infty$.
Define the kernel
\[
\kappa_\psi(x):=\E\!\left[\psi(x^\gamma Z_\gamma)\right],\qquad x>0.
\]
Then
\[
\E[\psi(Z)]=\int_0^\infty \kappa_\psi(x)\,F(dx).
\]
\end{lemma}

Lemma~\ref{lem:kernel_expectation_frechet} is conceptually important for design:
whenever the planner objective can be written as $\E[\psi(Z)]$ with $Z$ in product form,
the dependence on $F$ is through a linear functional with kernel $\kappa_\psi$.

\subsection[Differentiating the map F to H(gamma,F)]{Differentiating the map $F\mapsto H_{\gamma,F}$}\label{subsec:frechet_derivative}

Because $H_{\gamma,F}(x)$ is an integral of a fixed kernel against $F$,
the operator $T$ has an explicit variational derivative.
Let $\mathcal{M}([0,\infty))$ denote the space of finite signed measures on $[0,\infty)$.

\begin{proposition}[G\^ateaux derivative of the extreme value operator]\label{prop:gateaux}
Fix $\gamma$ and consider the map $T:F\mapsto H_{\gamma,F}$.
Let $\nu\in\mathcal{M}([0,\infty))$ satisfy $\nu([0,\infty))=0$ and $\int_0^\infty u\,\nu(du)=0$, and let $F_\varepsilon:=F+\varepsilon \nu$ be a perturbation for $\varepsilon$ in a neighborhood of $0$ such that $F_\varepsilon$ remains a probability measure.
Then for every $x$ with $v_\gamma(x)<\infty$,
\[
\left.\frac{d}{d\varepsilon} H_{\gamma,F_\varepsilon}(x)\right|_{\varepsilon=0}
=\int_{0}^{\infty} e^{-v_\gamma(x)u}\,\nu(du).
\]
\end{proposition}

\begin{remark}[Why the mean constraint is stated here]
The derivative formula itself is algebraically valid without the restriction $\int u\,\nu(du)=0$.
We impose it in the proposition because the economic domain studied in the paper fixes $\E[X]=1$, so admissible perturbations should remain on that slice.
\end{remark}

\begin{remark}[Link to entropic design]
Proposition~\ref{prop:gateaux} identifies the kernel that governs marginal changes in the distribution of extremes when the type distribution is perturbed.
In Section~\ref{sec:entropic}, the same kernel reappears in the entropy-regularized design problem whenever the planner objective is expressed through cdf levels or their smooth linear combinations.
\end{remark}

\subsection{Geodesics and convexity}\label{subsec:geodesics}

We now combine the one-dimensional canonical Wasserstein geodesics with convexity of the Laplace kernel.

\begin{proposition}[Convexity of heterogeneous extreme value laws along type geodesics]\label{prop:convexity_H_along_geodesic}
Let $F^0,F^1\in\mathcal{P}_1([0,\infty))$ with $\E[X]=1$ and let $(F^t)_{t\in[0,1]}$ be the canonical monotone Wasserstein geodesic between them.
Then for every $z\ge 0$, the map $t\mapsto P_0^{\,t}(z):=\int e^{-zx}\,F^t(dx)$ is convex on $[0,1]$.
Consequently, for every $x$ such that $v_\gamma(x)<\infty$, the map
\[
t\ \mapsto\ H_{\gamma,F^t}(x)=P_0^{\,t}\!\big(v_\gamma(x)\big)
\]
is convex on $[0,1]$.
\end{proposition}

Proposition~\ref{prop:convexity_H_along_geodesic} provides a disciplined way to compare extremes along canonical paths in the space of heterogeneity distributions.
It is also useful for identification and policy discussions:
when $F$ is estimated, geodesic perturbations provide a structured local sensitivity analysis.

\begin{remark}[Raw-space geodesics preserve the mean one slice]\label{rem:raw_geodesics_mean_preserving}
If $F^0$ and $F^1$ both satisfy $\int x\,F^k(dx)=1$, then the canonical monotone Wasserstein geodesic $(F^t)_{t\in[0,1]}$ in Proposition~\ref{prop:convexity_H_along_geodesic} also satisfies $\int x\,F^t(dx)=1$ for every $t$.
Indeed, in one dimension the mean equals the integral of the quantile function, and quantiles interpolate linearly along this path.
This mean-preserving property is specific to the raw type geometry used in this section and will contrast with the adapted geodesics of Section~\ref{sec:Wstab}.
\end{remark}

We also record a moment functional that will recur in heavy tailed applications.

\begin{proposition}[Geodesic convexity of negative power moments]\label{prop:neg_moment_geodesic_convex}
Fix $\rho>0$.
Let $F^0,F^1$ be supported on $[a,\infty)$ for some $a>0$, and let $(F^t)_{t\in[0,1]}$ be their canonical monotone Wasserstein geodesic.
Then the map
\[
t\ \mapsto\ \int x^{-\rho}\,F^t(dx)
\]
is convex on $[0,1]$.
\end{proposition}

\begin{remark}[Where this enters later]
In Fr\'echet settings, many tail sensitive functionals can be written in terms of power moments of $X$ through the product representation in Proposition~\ref{prop:pushforward_rep} and the kernel representation in Lemma~\ref{lem:kernel_expectation_frechet}.
The convexity statements above alow one to control these objects along structured perturbations of $F$.
\end{remark}

\section{Wasserstein stability and coupling bounds}\label{sec:Wstab}

This section contains the genuinely geometric layer of the paper.
The Laplace-mixture representation already gives several pointwise and order-theoretic statements.
What transport adds is different: by coupling transformed types with a common exponential shock, we obtain quantitative control of whole induced laws, canonical ambient interpolation paths, and robustness statements with explicit moduli.
The key object remains the operator $T:F\mapsto H_{\gamma,F}$.
The substantive content is not merely that one may rewrite the type space after a monotone transform, but that the operator representation and the canonical coupling make this transformed geometry propagate into law-level Wasserstein bounds for heterogeneous extremes.
Complementary quantitative control of extreme-value approximation in Wasserstein-type distances has also recently been developed with Stein-type methods for Fr\'echet laws; see \citet{MansanarezPolySwan2025}.

\subsection{Wasserstein preliminaries and duality}\label{subsec:wass_prelim}

We work on $\R$ equipped with the metric $d(x,y)=|x-y|$.
For $p\ge 1$ and $\mu,\nu\in\cP_p(\R)$,
\[
W_p(\mu,\nu):=\left(\inf_{\pi\in\Gamma(\mu,\nu)} \int_{\R^2} |x-y|^p\,\pi(dx,dy)\right)^{1/p},
\]
where $\Gamma(\mu,\nu)$ is the set of couplings with marginals $(\mu,\nu)$.
In one dimension, $W_p$ admits the quantile representation, and for $p=1$ it admits the Kantorovich-Rubinstein dual formulation.
We collect the statements used below in Appendix~\ref{app:tools_wass}.

\subsection{A canonical coupling representation}\label{subsec:canonical_coupling}

Recall the canonical GEV parametrization from Section~\ref{subsec:evt_setup} and define
\[
v_\gamma(x):=-\log H_\gamma(x).
\]
Let $w_\gamma$ denote the inverse of $v_\gamma$ on its range.
For the canonical GEV family,
\[
w_\gamma(t)=
\begin{cases}
\dfrac{t^{-\gamma}-1}{\gamma}, & \gamma\neq 0,\\[0.6em]
-\log t, & \gamma=0,
\end{cases}
\qquad t>0.
\]

\begin{proposition}[Canonical representation of $H_{\gamma,F}$]\label{prop:canonical_rep}
Let $X\sim F$ and let $E\sim\mathrm{Exp}(1)$ be independent of $X$.
Define
\[
Z_{\gamma,F}:=w_\gamma\!\left(\frac{E}{X}\right).
\]
Then $Z_{\gamma,F}$ has cdf $H_{\gamma,F}$.

Equivalently,
\[
Z_{\gamma,F}=
\begin{cases}
\dfrac{X^\gamma E^{-\gamma}-1}{\gamma}, & \gamma\neq 0,\\[0.6em]
\log X-\log E, & \gamma=0,
\end{cases}
\]
and this random variable has cdf $H_{\gamma,F}$.
\end{proposition}

Proposition~\ref{prop:canonical_rep} reduces the construction of couplings for $H_{\gamma,F}$ to couplings on the type space.
Given any coupling of $X_1\sim F_1$ and $X_2\sim F_2$, one obtains a coupling of the corresponding extreme laws by using a common exponential variable $E$.

\subsection{Stability of the extreme value operator}\label{subsec:stab_HEV}

The canonical representation depends on $X$ through the transform $x\mapsto x^\gamma$ when $\gamma\neq 0$, and through $x\mapsto \log x$ when $\gamma=0$.
This motivates metrics on the type space that are adapted to the index $\gamma$.
For the geometric results in this section, the mean one normalization from Assumption~\ref{ass:mixedPoisson} is not needed.
We therefore temporarily enlarge the domain from the normalized economic type space to the ambient space of probability measures on $(0,\infty)$ for which the relevant transformed $p$th moment is finite.
That ambient domain is the natural one for the adapted metric and, crucially, the one on which the geodesic construction is closed.

\begin{definition}[Type metrics adapted to $\gamma$]\label{def:type_metrics}
Fix $p\ge 1$ and $\gamma\in\R$.
Let $s_\gamma:(0,\infty)\to\R$ be defined by
\[
s_\gamma(x):=
\begin{cases}
x^\gamma, & \gamma\neq 0,\\
\log x, & \gamma=0.
\end{cases}
\]
For $F_1,F_2$ on $(0,\infty)$ such that the pushforwards $F_i\circ s_\gamma^{-1}$ belong to $\cP_p(\R)$, define
\[
d_{\gamma,p}(F_1,F_2):=
W_p\!\left(F_1\circ s_\gamma^{-1},\,F_2\circ s_\gamma^{-1}\right).
\]
\end{definition}

\begin{proposition}[Domain of the adapted metric and behavior near zero]\label{prop:domain_dgamma}
Fix $p\ge 1$ and $\gamma\in\R$.
For probability measures $F_1,F_2$ on $(0,\infty)$, the distance $d_{\gamma,p}(F_1,F_2)$ is finite if and only if
\[
\int |s_\gamma(x)|^p\,F_i(dx)<\infty,\qquad i\in\{1,2\}.
\]
Equivalently:
\begin{enumerate}
\item if $\gamma>0$, one needs $\int x^{\gamma p}\,F_i(dx)<\infty$ for $i=1,2$;
\item if $\gamma=0$, one needs $\int |\log x|^p\,F_i(dx)<\infty$ for $i=1,2$;
\item if $\gamma<0$, one needs $\int x^{-|\gamma|p}\,F_i(dx)<\infty$ for $i=1,2$.
\end{enumerate}
In particular, mass near zero is harmless for $\gamma>0$ but can make $d_{\gamma,p}$ infinite for $\gamma\le 0$ even when an ordinary Wasserstein distance on the raw type space is finite.
\end{proposition}

\begin{remark}[Low-access types and the limits of adapted geometry]\label{rem:low_access_domain}
The transport theorem therefore remains available for very small access types as long as the relevant transformed moments are finite.
For $\gamma=0$ this is a logarithmic integrability requirement; for $\gamma<0$ it is an inverse-moment requirement of order $|\gamma|p$.
If an application permits heavier mass near zero than those conditions allow, the adapted metric is no longer the right comparison tool.
One should then revert to the Laplace-order statements, the raw space mean-preserving geodesics of Section~\ref{sec:geometry}, or the finite horizon and pointwise bounds in this section, none of which require $d_{\gamma,p}<\infty$.
\end{remark}

We now state the main stability bound.
The constants are explicit and reflect moment restrictions intrinsic to extreme-value tails.

\begin{theorem}[Wasserstein stability of $F\mapsto H_{\gamma,F}$]\label{thm:stability_HEV}
Fix $\gamma\in\R$ and $p\ge 1$.
Let $F_1,F_2$ be probability measures on $(0,\infty)$ such that $d_{\gamma,p}(F_1,F_2)<\infty$.

If $\gamma>0$, assume in addition that $p\gamma<1$.
If $\gamma\neq 0$, let $E\sim\mathrm{Exp}(1)$ and define $\Xi_\gamma:=E^{-\gamma}$.
Then
\[
W_p\!\left(H_{\gamma,F_1},H_{\gamma,F_2}\right)
\le
\frac{\left(\E[|\Xi_\gamma|^p]\right)^{1/p}}{|\gamma|}\, d_{\gamma,p}(F_1,F_2).
\]

If $\gamma=0$, then
\[
W_p\!\left(H_{0,F_1},H_{0,F_2}\right)\le d_{0,p}(F_1,F_2).
\]
\end{theorem}

\begin{remark}[What the transport layer adds]\label{rem:what_transport_adds}
The transform $s_\gamma$ by itself is only a reparameterization of types.
The substantive content of Theorem~\ref{thm:stability_HEV} is that, under the canonical coupling $Z_{\gamma,F}=w_\gamma(E/X)$, optimal transport on $s_\gamma(X)$ yields an explicit Wasserstein bound on the entire induced extreme law, with constant determined by the tail index through $\E[E^{-p\gamma}]$ in the Fr\'echet regime.
Corollary~\ref{cor:adapted_geodesic_stability} then propagates this to canonical ambient paths in law space.
These quantitative whole law statements do not follow from the Laplace mixture alone.
\end{remark}

\begin{remark}[Moment restrictions are intrinsic]\label{rem:moment_intrinsic}
For $\gamma>0$ one has $\E[E^{-p\gamma}]=\Gamma(1-p\gamma)$, so $\E[|\Xi_\gamma|^p]<\infty$ holds if and only if $p\gamma<1$.
For $\gamma<0$, by contrast, $\Xi_\gamma=E^{|\gamma|}$ has finite moments of every order.
The heavy-tail restriction is therefore not an artifact of the coupling argument.
It reflects that the Fr\'echet-type limit law has finite $p$th moment only when $p<1/\gamma$, and therefore $W_p(H_{\gamma,F_1},H_{\gamma,F_2})$ is only defined in the usual Wasserstein sense in that regime.
\end{remark}

Theorem~\ref{thm:stability_HEV} is proved by applying the random Lipschitz contraction Lemma~\ref{lem:random_lipschitz} to the representation in Proposition~\ref{prop:canonical_rep}, using a coupling of $s_\gamma(X_1)$ and $s_\gamma(X_2)$ and a common exponential variable.
The construction is in the same spirit as the coupling arguments emphasized by \citet{BobbiaDombryVarron2019}, but here the comparison is between heterogeneous extreme laws generated by different opportunity distributions.

The theorem becomes most useful once the adapted metric is treated as a geometry in its own right.
Because $d_{\gamma,p}$ is an ordinary Wasserstein distance after the transform $s_\gamma$, it automatically induces canonical constant-speed geodesics in the ambient transformed type space.
Along those paths, the entire law of extremes moves in a Lipschitz way.
This is the main result in the paper that specifically requires the Wasserstein geometry rather than only the Laplace form.

\begin{corollary}[Geodesic control in the adapted type geometry]\label{cor:adapted_geodesic_stability}
Fix $\gamma\in\R$ and $p\ge 1$.
Let $F^0,F^1$ be two probability measures on $(0,\infty)$ for which the assumptions of Theorem~\ref{thm:stability_HEV} hold.
Let $G^0:=s_\gamma\#F^0$ and $G^1:=s_\gamma\#F^1$, and let $(G^t)_{t\in[0,1]}$ be the canonical monotone constant-speed $W_p$ geodesic from $G^0$ to $G^1$.
Define $F^t:=(s_\gamma^{-1})\#G^t$.
Then, for all $s,t\in[0,1]$,
\[
d_{\gamma,p}(F^s,F^t)=|t-s|\,d_{\gamma,p}(F^0,F^1).
\]
Moreover,
\[
W_p\!\left(H_{\gamma,F^s},H_{\gamma,F^t}\right)
\le C_{\gamma,p}\,|t-s|\,d_{\gamma,p}(F^0,F^1),
\]
where $C_{\gamma,p}$ denotes the constant from Theorem~\ref{thm:stability_HEV}.
\end{corollary}

\begin{remark}[Ambient geodesics versus the mean one normalization]\label{rem:ambient_vs_meanone}
The path in Corollary~\ref{cor:adapted_geodesic_stability} is a geodesic in the ambient transformed type space, not necessarily in the mean one slice used in the economic model.
Indeed, $Q_{G^t}(u)$ is affine in $t$ in transformed coordinates, but $\int s_\gamma^{-1}(Q_{G^t}(u))\,du$ is generally not constant unless $s_\gamma^{-1}$ is affine, which occurs only when $\gamma=1$.
Thus the adapted path should be read as a geometric comparison device, not as a literal mean-preserving policy path.
If one instead wants exact mean preservation at every step, the raw space geodesics of Section~\ref{subsec:geodesics} remain available, but they need not inherit the constant-speed statement for $d_{\gamma,p}$.
\end{remark}

\begin{proposition}[Renormalization bridge back to the mean one slice]\label{prop:renormalization_bridge}
Fix $\gamma\in\R$ and $p\ge 1$.
Let $F$ be a probability measure on $(0,\infty)$ with finite mean $m(F):=\int x\,F(dx)\in(0,\infty)$ and finite transformed $p$th moment.
Define the mean one renormalization
\[
\widetilde F:=(x\mapsto x/m(F))_\#F.
\]
Then
\[
d_{\gamma,p}(F,\widetilde F)=
\begin{cases}
\left|1-m(F)^{-\gamma}\right|\left(\int x^{\gamma p}\,F(dx)\right)^{1/p}, & \gamma\neq 0,\\[0.8em]
\left|\log m(F)\right|, & \gamma=0.
\end{cases}
\]
If, in addition, the assumptions of Theorem~\ref{thm:stability_HEV} hold for $F$ and $\widetilde F$, then
\[
W_p\!\left(H_{\gamma,F},H_{\gamma,\widetilde F}\right)
\le
C_{\gamma,p}\,d_{\gamma,p}(F,\widetilde F),
\]
where $C_{\gamma,p}$ is the modulus from Theorem~\ref{thm:stability_HEV}.
In particular, if $m_t:=\int x\,F^t(dx)$ and $\widetilde F^t:=(x\mapsto x/m_t)_\#F^t$ along the ambient geodesic from Corollary~\ref{cor:adapted_geodesic_stability}, the same bound applies pointwise in $t$.
\end{proposition}

\begin{corollary}[Quantile-schedule control of the induced extreme law]\label{cor:quantile_schedule_stability}
Under the assumptions of Theorem~\ref{thm:stability_HEV}, let $Q_{\gamma,F_i}:(0,1)\to\R$ denote the quantile function of the law $H_{\gamma,F_i}$ for $i\in\{1,2\}$.
Then
\[
\left(\int_0^1 \left|Q_{\gamma,F_1}(u)-Q_{\gamma,F_2}(u)\right|^p\,du\right)^{1/p}
=
W_p\!\left(H_{\gamma,F_1},H_{\gamma,F_2}\right)
\le
C_{\gamma,p}\,d_{\gamma,p}(F_1,F_2).
\]
Along the adapted geodesic from Corollary~\ref{cor:adapted_geodesic_stability}, this yields
\[
\|Q_{\gamma,F^s}-Q_{\gamma,F^t}\|_{L^p(0,1)}
\le
C_{\gamma,p}\,|t-s|\,d_{\gamma,p}(F^0,F^1).
\]
\end{corollary}

An immediate benchmark implication is obtained by comparing $F$ to $\delta_1$.
Whenever the conditions of Theorem~\ref{thm:stability_HEV} hold,
\[
W_p\!\left(H_{\gamma,F},H_\gamma\right)\le C_{\gamma,p}\,d_{\gamma,p}(F,\delta_1),
\]
where $C_{\gamma,p}$ denotes the constant from Theorem~\ref{thm:stability_HEV}.
Combined with Proposition~\ref{prop:type_metric_controls} in the regimes covered there, this translates metric misalocation on the type space into an explicit distortion bound for extreme outcomes relative to the homogeneous benchmark.

\begin{corollary}[Pointwise stability of the cdf]\label{cor:pointwise_cdf_stability}
Fix $\gamma\in\R$ and let $F_1,F_2\in\cP_1([0,\infty))$.
Then for every $x$ in the interior of the support of $H_\gamma$,
\[
\left|H_{\gamma,F_1}(x)-H_{\gamma,F_2}(x)\right|
\le
v_\gamma(x)\,W_1(F_1,F_2).
\]
\end{corollary}

\begin{remark}[Pointwise versus geometric bounds]\label{rem:pointwise_vs_geometric}
Corollary~\ref{cor:pointwise_cdf_stability} follows directly from the Laplace representation
$H_{\gamma,F}(x)=\int e^{-v_\gamma(x)z}\,F(dz)$ and the fact that $z\mapsto e^{-v_\gamma(x)z}$ is $v_\gamma(x)$-Lipschitz on $[0,\infty)$.
It is therefore not itself a geometric theorem.
Theorem~\ref{thm:stability_HEV}, Proposition~\ref{prop:renormalization_bridge}, and Corollary~\ref{cor:quantile_schedule_stability} are the results that genuinely use the adapted Wasserstein geometry. Those statements control whole induced laws, the bridge back to the mean one slice, and the full quantile schedule. They are not recovered from fixed threshold cdf bounds without additional density information.
\end{remark}

In many applications one also wants conditions under which $d_{\gamma,p}$ is controlled by a standard Wasserstein distance on the type space.
We record sufficient conditions that isolate the role of behavior near $0$.

\begin{proposition}[Relating adapted type metrics to standard Wasserstein distances]\label{prop:type_metric_controls}
Let $p\ge 1$ and let $F_1,F_2\in\cP_p([0,\infty))$.

\paragraph{Case 1: $0<\gamma\le 1$.}
Then
\[
d_{\gamma,p}(F_1,F_2)\le W_p(F_1,F_2)^\gamma.
\]

\paragraph{Case 2: $\gamma<0$.}
Assume $F_1$ and $F_2$ are supported on $[a,\infty)$ for some $a>0$.
Then
\[
d_{\gamma,p}(F_1,F_2)\le |\gamma|\,a^{\gamma-1}\,W_p(F_1,F_2).
\]

\paragraph{Case 3: $\gamma=0$.}
Assume $F_1$ and $F_2$ are supported on $[a,\infty)$ for some $a>0$.
Then
\[
d_{0,p}(F_1,F_2)\le a^{-1}\,W_p(F_1,F_2).
\]
\end{proposition}

\begin{remark}[Why we keep $d_{\gamma,p}$ in the main bound]
In the positive-tail regime covered by Theorem~\ref{thm:stability_HEV}, namely $0<\gamma<1/p$, the map $x\mapsto x^\gamma$ is not globally Lipschitz on $[0,\infty)$.
For $0<\gamma<1$, the obstruction is behavior near zero.
Consequently, a linear stability bound stated directly in terms of $W_p(F_1,F_2)$ would either require extra support restrictions or lose the natural linear geometry of the canonical coupling.
The adapted metric $d_{\gamma,p}$ is therefore the correct primitive quantity in the main bound, while Proposition~\ref{prop:type_metric_controls} records conditions under which it can be controlled by a standard Wasserstein distance.
\end{remark}

\subsection{Stability of extreme outcome functionals}\label{subsec:stab_functionals}

The Wasserstein bounds above immediately imply stability of a large class of tail-sensitive functionals.

\begin{proposition}[Lipschitz functionals of the limit law]\label{prop:lip_functionals_limit}
Let $\psi:\R\to\R$ be $L$-Lipschitz.
Let $Z_i\sim H_{\gamma,F_i}$ for $i\in\{1,2\}$ and assume $W_1(H_{\gamma,F_1},H_{\gamma,F_2})<\infty$.
Then
\[
\left|\E[\psi(Z_1)]-\E[\psi(Z_2)]\right|
\le
L\,W_1\!\left(H_{\gamma,F_1},H_{\gamma,F_2}\right).
\]
Combining with Theorem~\ref{thm:stability_HEV} yields explicit bounds in terms of $d_{\gamma,p}(F_1,F_2)$ whenever the conditions of that theorem hold.
\end{proposition}

\subsection{Finite horizon maxima and tail probability stability}\label{subsec:finite_theta}

Fix $\theta>0$ and let $M_\theta$ be the maximum of $N_\theta$ i.i.d.\ draws from $G$, where $N_\theta\mid X\sim\mathrm{Poisson}(\theta X)$ as in Section~\ref{sec:env}.
Conditional on $X$, one has
\[
\Pr(M_\theta\le x\mid X)=\exp\!\big(-\theta X(1-G(x))\big),
\]
and therefore $\Pr(M_\theta\le x)=\E[\exp(-\theta X(1-G(x)))]$.

\begin{proposition}[Pointwise stability for finite $\theta$]\label{prop:finite_theta_stability}
Let $F_1,F_2\in\cP_1([0,\infty))$.
Then for every $x\in\R$,
\[
\left|\Pr_{F_1}(M_\theta\le x)-\Pr_{F_2}(M_\theta\le x)\right|
\le
\theta\,(1-G(x))\,W_1(F_1,F_2).
\]
\end{proposition}

\begin{remark}[Finite-horizon versus asymptotic robustness]\label{rem:finite_vs_asymptotic}
Proposition~\ref{prop:finite_theta_stability} is a finite horizon robustness statement.
It does not use the extreme-value limit.
By contrast, Theorem~\ref{thm:second_order_expansion} below is an asymptotic expansion around the limit law.
Keeping these two layers separate is important in empirical work: finite-sample prediction error and EVT approximation error are different objects.
\end{remark}

\subsection{Rates under second-order tail conditions}\label{subsec:rates}

We conclude with a second-order expansion for the prelimit cdf of normalized maxima.
The proof is given in Appendix~A.
In this subsection, unlike the first-order limit theorem in Section~\ref{subsec:hev_limit}, we work with continuous normalizing functions $a:(1,\infty)\to(0,\infty)$ and $b:(1,\infty)\to\R$.
This is the standard continuous-parameter formulation in extreme value theory and avoids the artificial $O(1/\theta)$ oscillation created by the step extension $a_{\lfloor\theta\rfloor},b_{\lfloor\theta\rfloor}$.
Assumption~\ref{ass:second_order} is a standard second-order regularity condition; see, for example, \citet[Chapters~2.10 and~9]{deHaanFerreira2006}.

\begin{assumption}[Second-order tail condition]\label{ass:second_order}
In addition to Assumption~\ref{ass:doa}, suppose there exist continuous normalizing functions $a:(1,\infty)\to(0,\infty)$ and $b:(1,\infty)\to\R$, a scalar function $A(\theta)\to 0$, and a locally bounded function $h_\gamma$ such that
\[
\theta\big(1-G(b(\theta)+a(\theta)x)\big)=v_\gamma(x) + A(\theta)\,h_\gamma(x) + o\!\big(A(\theta)\big)
\]
as $\theta\to\infty$, locally uniformly in $x$ on the interior of the support of $H_\gamma$.
\end{assumption}

\begin{theorem}[Second-order expansion for heterogeneous extremes]\label{thm:second_order_expansion}
Suppose Assumptions~\ref{ass:mixedPoisson} and~\ref{ass:second_order} hold and let
\[
Z_\theta:=\frac{M_\theta-b(\theta)}{a(\theta)}.
\]
Fix a compact set $K$ contained in the interior of the support of $H_\gamma$.
Then there exists a remainder $r_{\theta,K}$ on $K$ such that
\[
\sup_{x\in K}\frac{|r_{\theta,K}(x)|}{|A(\theta)|}\to 0
\qquad\text{as }\theta\to\infty,
\]
and, uniformly for $x\in K$,
\[
\Pr(Z_\theta\le x)=H_{\gamma,F}(x) - A(\theta)\,h_\gamma(x)\,\E\!\left[X e^{-v_\gamma(x)X}\right] + r_{\theta,K}(x).
\]
\end{theorem}

\begin{remark}[Interpretation]
The leading correction term is multiplicatively separable.
The function $h_\gamma$ captures the second-order tail behavior of $G$, while $x\mapsto \E[X e^{-v_\gamma(x)X}]$ is the heterogeneity kernel inherited from the Laplace mixture.
This separation is what allows one to distinguish EVT approximation error from heterogeneity-driven sensitivity.
Converting the expansion into global $W_1$ rates require additional tail-envelope conditions and we therefore left outside the present theorem.
\end{remark}

\section{A complementary entropy projection for heterogeneity}\label{sec:entropic}

In this section we study a tractable normative problem on the same type space used in the positive analysis.
A planner reallocates opportunities across agents to improve a tail-sensitive criterion while penalizing deviations from a baseline distribution by relative entropy.
The object is deliberately modest: it is a one-marginal entropy projection, not a full two-marginal transport problem, and it does not rely on the Wasserstein geometry developed earlier.
It is useful in our framework, but the connection should be described accurately: this section is complementary rather than foundational.
The structural link to the positive analysis is that the same Laplace kernel that defines $F\mapsto H_{\gamma,F}$ also generates the marginal score functions that enter the planner's first-order conditions for cdf criteria and, through the stochastic representation below, for expected utility of normalized extremes.

\subsection{A variational planner problem}\label{subsec:planner_problem}

Fix a baseline distribution $F_0$ on $(0,\infty)$ such that $\E_{F_0}[X]=1$.
We interpret $F_0$ as the status quo distribution of effective search intensities, network degrees, or other shifters of arrival rates.
A policy induces a new distribution $F$, and we write $X\sim F$.

For $F\ll F_0$, define the relative entropy
\[
D_{\mathrm{KL}}(F\|F_0):=\int \log\!\left(\frac{dF}{dF_0}\right)\,dF,
\]
with the convention $D_{\mathrm{KL}}(F\|F_0)=+\infty$ if $F$ is not absolutely continuous with respect to $F_0$.

Let $H_{\gamma,F}$ denote the heterogeneous extreme-value law from Proposition~\ref{prop:HEV}.
We focus on tail-sensitive objectives of the form
\[
U(F)=\E\big[u(Z)\big],\qquad Z\sim H_{\gamma,F},
\]
for a measurable function $u$ such that the expectation is well defined under the candidate distributions considered below.
This class includes expected utility of normalized extremes and many smoothed tail criteria.
We comment on high-quantile objectives at the end of the section.

The planner solves
\begin{equation}\label{eq:design}
\sup_{F:\ \E_F[X]=1}\ \Big\{U(F) - \lambda\,D_{\mathrm{KL}}(F\|F_0)\Big\},
\end{equation}
where $\lambda>0$ governs the marginal cost of reshaping opportunities away from the baseline.

\begin{remark}[Interpretation]
The constraint $\E_F[X]=1$ normalizes the aggregate scale of opportunities.
In a network interpretation it fixes total expected offer arrival in the population and alows the planner to reallocate the cross-sectional distribution.
The entropy term penalizes concentrated reallocations and yields a strictly concave problem whenever $U$ is linear in $F$, which is the main case treated below.
\end{remark}

\begin{remark}[Normalized welfare and prelimit maxima]\label{rem:normalized_welfare}
The criterion $U(F)=\E[u(Z)]$ should be read as the asymptotic counterpart of prelimit preferences over maxima.
Indeed, if $u$ is bounded and continuous, then Proposition~\ref{prop:HEV} implies
\[
\E\!\left[u\!\left(\frac{M_\theta-b_\theta}{a_\theta}\right)\right]\to U(F)
\qquad\text{as }\theta\to\infty.
\]
Equivalently, the planner may be viewed as evaluating a sequence of prelimit utilities $u_\theta(m):=u((m-b_\theta)/a_\theta)$ that preserves the affine normalization used by extreme-value theory.
\end{remark}

\subsection{A stochastic representation and linearization of expected utility}\label{subsec:frechet_linearization}

The canonical representation from Section~\ref{subsec:canonical_coupling} immediately linearizes expected-utility objectives in the heterogeneity distribution.
Let $w_\gamma:=v_\gamma^{-1}$ denote the inverse map introduced there.

\begin{lemma}[Linearization of expected utility]\label{lem:exp_rep}
Let $X\sim F$ and let $E\sim\mathrm{Exp}(1)$ be independent.
Define
\[
Z:=w_\gamma\!\left(\frac{E}{X}\right),
\]
so that $Z\sim H_{\gamma,F}$ by Proposition~\ref{prop:canonical_rep}.
If $u:\R\to\R$ is measurable and $\E\big[|u(Z)|\big]<\infty$, then
\[
U(F)=\E\big[u(Z)\big]=\int \psi_u(x)\,F(dx),
\qquad
\psi_u(x):=\E\!\left[u\!\left(w_\gamma\!\left(\frac{E}{x}\right)\right)\right].
\]
\end{lemma}

\begin{proof}
The distributional representation of $Z$ is Proposition~\ref{prop:canonical_rep}.
The formula for $U(F)$ then follows by iterated expectation.
\end{proof}

\begin{remark}[Fr\'echet specialization]
When $\gamma>0$, Proposition~\ref{prop:canonical_rep} reduces to the product representation from Section~\ref{subsec:pareto_rep}: if $\Xi_\gamma:=E^{-\gamma}$, then
\[
1+\gamma Z=X^\gamma \Xi_\gamma.
\]
\end{remark}

\subsection{Duality and exponential tilting}\label{subsec:duality}

We now solve \eqref{eq:design} in the canonical case where the objective is linear in $F$.
This case is both the baseline design problem and the relevant reduction for expected-utility objectives by Lemma~\ref{lem:exp_rep}.
The exponential-tilt form itself is the standard entropy-projection conclusion.
What is specific here is that the heterogeneous-EVT operator supplies the score $\psi$ and reduces the mean constraint to a one-dimensional dual variable.
The duality argument is recorded in Appendix~\ref{app:tools_entropic}; see also \citet{DonskerVaradhan1975}, \citet{Leonard2014}, \citet{GhosalNutzBernton2022}, and \citet{Nutz2022}, for background on the broader Schr\"odinger and entropic-transport literature.

\begin{assumption}[Linear objective, local integrability, and interiority]\label{ass:linear_objective}
There exists a measurable function $\psi:(0,\infty)\to\R$ such that
\[
U(F)=\int \psi(x)\,F(dx)
\quad\text{for all }F\text{ with }\E_F[X]=1\text{ and }D_{\mathrm{KL}}(F\|F_0)<\infty.
\]
For $\eta\in\R$, define
\[
Z(\eta):=\int \exp\!\left(\frac{\psi(x)+\eta x}{\lambda}\right)\,F_0(dx)
\]
and, for $k\in\{1,2\}$,
\[
M_k(\eta):=\int x^k\exp\!\left(\frac{\psi(x)+\eta x}{\lambda}\right)\,F_0(dx).
\]
Assume there exists a nonempty open interval $\mathcal D\subset\R$ such that, for all $\eta\in\mathcal D$,
\[
Z(\eta)<\infty,\qquad M_1(\eta)<\infty,\qquad M_2(\eta)<\infty,
\qquad
\int |\psi(x)|\exp\!\left(\frac{\psi(x)+\eta x}{\lambda}\right)\,F_0(dx)<\infty.
\]
Define
\[
m(\eta):=\frac{M_1(\eta)}{Z(\eta)}.
\]
Assume moreover that there exist $\eta_-,\eta_+\in\mathcal D$ with
\[
m(\eta_-)<1<m(\eta_+).
\]
\end{assumption}

\begin{theorem}[Strong duality and exponential tilt]\label{thm:entropic_duality}
Suppose Assumption~\ref{ass:linear_objective} holds.
Then there exists a unique scalar $\eta^\star\in\mathcal D$ satisfying $m(\eta^\star)=1$.
The design problem \eqref{eq:design} admits a unique optimizer $F^\star$, and $F^\star$ is given by the exponential tilt
\begin{equation}\label{eq:tilt}
\frac{dF^\star}{dF_0}(x)
=
\frac{\exp\!\left(\frac{\psi(x)+\eta^\star x}{\lambda}\right)}
{\int \exp\!\left(\frac{\psi(t)+\eta^\star t}{\lambda}\right)\,F_0(dt)}.
\end{equation}
Moreover, $\E_{F^\star}[X]=1$.

The optimal value of \eqref{eq:design} admits the dual representation
\begin{equation}\label{eq:dual_eta}
\sup_{F:\E_F[X]=1}\Big\{U(F)-\lambda D_{\mathrm{KL}}(F\|F_0)\Big\}
=
\inf_{\eta\in\mathcal D}\left\{
\lambda \log\!\int \exp\!\left(\frac{\psi(x)+\eta x}{\lambda}\right)\,F_0(dx)
-\eta
\right\}.
\end{equation}
\end{theorem}

\paragraph{Comment.}
Equation \eqref{eq:tilt} is the key structural conclusion of the section.
Relative to $F_0$, the planner reweights types with higher score $\psi(x)$, with the magnitude of reweighting disciplined by $\lambda$.
The scalar $\eta^\star$ is the Lagrange multiplier on the mean one constraint and is pinned down by a single scalar equation $m(\eta)=1$.
With the sign convention in \eqref{eq:dual_eta}, the shadow value of relaxing the target mean upward is $-\eta^\star$.
The novelty here is not the Gibbs form by itself, but the identification of the heterogeneous-EVT score $\psi$ and the explicit one-dimensional implementation of the mean constraint.
The interiority assumption in Assumption~\ref{ass:linear_objective} is the precise condition that guarantees existence of a finite dual optimizer.
In many applications it is verified through a standard steepness argument for the tilted mean $m(\eta)$.

A complete proof of Theorem~\ref{thm:entropic_duality} is given in Appendix~A.
It applies the entropy-duality tools in Appendix~\ref{app:tools_entropic}, but the proof is written directly for the one-dimensional mean constraint used here so that the existence of the dual optimizer is explicit rather than implicit.

\subsection{First-order conditions and the link to the heterogeneous EVT kernel}\label{subsec:connect_A_C}

The design problem inherits its structure from the heterogeneous extreme-value operator
\[
H_{\gamma,F}(x)=\int e^{-v_\gamma(x)z}\,F(dz).
\]
Section~\ref{subsec:frechet_derivative} records the corresponding directional derivative,
\[
\delta H_{\gamma,F}(x)=\int e^{-v_\gamma(x)z}\,\delta F(dz),
\]
which makes explicit that the Laplace kernel $z\mapsto e^{-v_\gamma(x)z}$ is the marginal channel through which perturbations of $F$ affect the distribution of extremes.

In the linear class treated in Theorem~\ref{thm:entropic_duality}, the score function $\psi$ is the marginal value of shifting mass toward type $x$.
For expected-utility objectives $U(F)=\E[u(Z)]$, Lemma~\ref{lem:exp_rep} identifies this marginal value explicitly as $\psi_u(x)=\E[u(w_\gamma(E/x))]$.
This is the direct bridge to Section~\ref{sec:Wstab}: the same stochastic representation that yields linearity for design also provides the canonical coupling device for stability.

\begin{remark}[Direct cdf criteria use the EVT kernel literally]\label{rem:direct_kernel_criteria}
For any fixed threshold $y$, the criterion
\[
U_y(F):=H_{\gamma,F}(y)
\]
is already linear in $F$, with score
\[
\psi_y(x)=e^{-v_\gamma(y)x}.
\]
More generally, any weighted average of cdf levels of the form $\int H_{\gamma,F}(y)\,\nu(dy)$ has score $x\mapsto \int e^{-v_\gamma(y)x}\,\nu(dy)$ whenever the integral is finite.
This is the cleanest sense in which the same kernel governs both the heterogeneous extreme-value operator and a nontrivial class of normative objectives.
\end{remark}

\subsection{Closed-form solutions for canonical tail objectives}\label{subsec:closed_form}

Closed-form solutions are especially transparent when $\psi$ is a simple transform of $x$.
This class is relevant economically because several extreme outcome statistics in heterogeneous EVT reduces to moments or inverse moments of $X$ under Fr\'echet-type normalization, as in \citet{Mangin2025}.
The formulas below should be read together with their admissibility conditions: the exponential tilt is meaningful only when Assumption~\ref{ass:linear_objective} is satisfied.

\paragraph{Power and inverse-power objectives.}
Let $\rho>0$ and consider objectives of the form
\[
U(F)=C\cdot \E_F\!\big[X^{-\rho}\big]
\quad\text{or}\quad
U(F)=C\cdot \E_F\!\big[X^{\rho}\big],
\]
for a constant $C\in\R$.
These correspond to $\psi(x)=C x^{-\rho}$ or $\psi(x)=C x^\rho$ in Assumption~\ref{ass:linear_objective}.
When the relevant integrability conditions hold, Theorem~\ref{thm:entropic_duality} yields
\[
\frac{dF^\star}{dF_0}(x)\propto \exp\!\left(\frac{C x^{-\rho}+\eta^\star x}{\lambda}\right)
\quad\text{or}\quad
\frac{dF^\star}{dF_0}(x)\propto \exp\!\left(\frac{C x^{\rho}+\eta^\star x}{\lambda}\right),
\]
with $\eta^\star$ chosen so that $\E_{F^\star}[X]=1$.

\paragraph{Admissibility conditions.}
For inverse-power scores $x\mapsto Cx^{-\rho}$, a sufficient condition is that the baseline support be bounded away from zero, say $\mathrm{supp}(F_0)\subseteq [a,\infty)$ with $a>0$; without such a condition, the case of the positive coefficient $C>0$ typically makes the normalizing integral diverge near zero.
For positive-power scores $x\mapsto Cx^{\rho}$, admissibility depends on the right tail of $F_0$.
A sufficient condition is that $F_0$ be light enough so that
\[
\int \exp\!\left(\frac{C x^{\rho}+\eta x}{\lambda}\right)\,F_0(dx)<\infty
\]
for all $\eta$ in some open interval containing the optimizer.
When $\rho>1$ and $C>0$, this excludes exponential or heavier baseline tails.
These support and tail restrictions are not technical decoration; they are exactly what ensures the closed-form tilt is well defined.

\paragraph{Implementation.}
In applications, $\eta^\star$ is obtained by solving the scalar equation $m(\eta)=1$.
The value function is then obtained from the dual expression \eqref{eq:dual_eta}.
The entire computation is one-dimensional once the score $\psi$ is known.

\subsection{Interpretation: policy as reallocation of opportunities}\label{subsec:policy_interp}

The entropy penalty $D_{\mathrm{KL}}(F\|F_0)$ imposes a reduced-form cost of reallocating opportunities across types.
In a labor market network interpretation, $F$ can for instance represent the distribution of effective meeting rates or centralities, and policy changes referral frictions, search subsidies, or platform rules that alter the induced $F$.
The entropy term summarizes implementation frictions that render highly concentrated reallocations costly.
In richer environments, one can interpret the exponential tilt as the reduced-form allocation induced by subsidy, tax, or referral policies that change the private return to search intensity.
It is beyond the scope of this paper to study decentralized implementation, but we note that the one-dimensional form of \eqref{eq:tilt} makes that question especially tractable.

A natural extension, which we keep separate from the core analysis to preserve the one-dimensional tractability of \eqref{eq:design}, replaces the marginal penalty by a genuine entropic transport problem on a richer type space with cost $c(x_0,x)$.
That extension would produce a pair of dual potentials on the source and target type spaces.
The present section should therefore be read as an entropy projection on the target marginal, not as a full Schr\"odinger bridge.

\paragraph{Quantile objectives.}
If the objective is a high quantile, $U(F)=\mathrm{Q}_{1-\alpha}(H_{\gamma,F})$, then $U$ is typically not linear in $F$.
Two tractable approaches are to work with smoothed tail criteria or tail expectations that admit linear representations via Lemma~\ref{lem:exp_rep}, or to characterize optimizers through first-order conditions using the directional derivative of $F\mapsto H_{\gamma,F}$ from Section~\ref{subsec:frechet_derivative}.
We implement the linear class in the core text, leaving exact quantile design as a promising extension.

\section{Application: labor market networks and the distribution of top wages}\label{sec:apps}

This section illustrates how heterogeneous extreme-value theory can be used as a reduced-form device in labor market settings where access to job opportunities is mediated by social networks.
The key object is the cross-sectional distribution $F$ of effective offer-arrival intensities induced by heterogeneity in network position, referrals, and related opportunity channels.
Once $F$ is specified or measured, the general results in Sections~\ref{sec:geometry}, \ref{sec:Wstab}, and \ref{sec:entropic} translate into three types of statements:
(i) comparative statics for the right tail of wages as a function of network inequality,
(ii) robustness bounds for tail predictions under measurement or estimation error in $F$, and
(iii) an entropy-regularized policy problem that reallocates opportunities subject to implementation frictions.

While the broader labor network motivation was discussed in Section~\ref{subsec:literature}, here our goal is to translate the general results into a familiar reduced-form environment for tail wage outcomes.
Throughout, the object remains the cross-sectional law of a randomly drawn worker's normalized maximum under a common offer distribution.
Remark that we do not seek to model herein equilibrium wage setting, endogenous search, or welfare.
Accordingly, statements below about segregation, homophily, or policy should be read as claims about how those forces reshape the reduced-form access distribution $F$, not as standalone equilibrium or welfare conclusions.

\subsection{A stylized network-based search environment}\label{subsec:net_env}

\paragraph{Types and offer arrivals.}
Consider a large population of workers indexed by $i$.
Each worker has a network position summarized by a scalar type $X_i>0$ capturing effective access to job opportunities.
The interpretation of $X_i$ is as a shifter of job offer arrival rates, generated by a referral network, local contacts, or platform-mediated matching.
We impose the normalization $\E[X_i]=1$, so the average scale of opportunities is fixed and the object of interest is the cross-sectional distribution of access.

Fix a market thickness or horizon parameter $\theta>0$.
Conditional on $X_i$, the number of offers received by worker $i$ follows the mixed-Poisson specification from Section~\ref{sec:env}:
\[
N_{i,\theta}\mid X_i \sim \mathrm{Poisson}(\theta X_i).
\]
Let $F$ denote the cross-sectional distribution of $X_i$ in the population, so $X\sim F$ for a randomly drawn worker.

\paragraph{Offer values and realized wages.}
Let $\{Y_{ij}\}_{j\ge 1}$ be i.i.d.\ wage offers with common distribution $G$, independent of $(X_i,N_{i,\theta})$.
Worker $i$ accepts the best offer over the horizon,
\[
M_{i,\theta}:=\sup_{1\le j\le N_{i,\theta}} Y_{ij},
\]
with the convention $\sup\emptyset=-\infty$.
The economic content is that heterogeneity enters through access to opportunities, while the distribution of offer values is common across workers.

\paragraph{From network primitives to $F$.}
The distribution $F$ can be linked to standard network objects.
As a benchmark, suppose workers are nodes in an undirected contact network and each neighbor generates job leads at approximately independent random times.
Conditional on degree $D_i$, the number of leads over a horizon is then well approximated by a Poisson random variable with mean proportional to $\theta D_i$.
After normalization, one can set $X_i:=D_i/\E[D]$, so that $\E[X_i]=1$ and $F$ is the distribution of normalized degrees.
With weighted links, one can take $X_i$ as normalized weighted degree.
If the arrival mechanism aggregates information over longer paths, $X_i$ can represent a normalized centrality index.
The analysis that follows uses $F$ as the reduced-form sufficient statistic for the cross-sectional opportunity structure.

\paragraph{Tail normalization and the heterogeneous extreme-value law.}
Assume $G$ is in the max domain of attraction of the generalized extreme-value law $H_\gamma$ with index $\gamma$.
Under the normalization in Section~\ref{subsec:evt_setup}, the distribution of normalized outcomes for a randomly drawn worker converges to the heterogenous extreme-value law from Proposition~\ref{prop:HEV}:
\[
H_{\gamma,F}(x)=\E\!\left[\exp\!\big(-X\,v_\gamma(x)\big)\right],
\qquad X\sim F.
\]
Thus, once $F$ is specified or estimated, the model delivers a tractable mapping from network heterogeneity into the cross-sectional distribution of workers' top wages in the sense of maxima over offers.

\subsection{Positive implications: network heterogeneity and the right tail}\label{subsec:positive}

\paragraph{Network inequality and tail losses under convex order.}
A natural notion of increased network inequality is a mean-preserving spread of $F$ holding $\E[X]=1$ fixed, equivalently an increase in the convex order; see, e.g., \citet{RothschildStiglitz1970} and \citet{ShakedShanthikumar2007}.
In the present setting, this order has a direct implication for extremes because, for each $z>0$, the kernel $x\mapsto \exp(-zx)$ is convex and decreasing.
This is precisely the mechanism behind the Jensen and Laplace-ordering results in Section~\ref{subsec:order}.

\begin{corollary}[Network inequality and normalized top wages]\label{cor:net_inequality}
Fix $\gamma$ and let $F_1,F_2$ satisfy $\E[X]=1$.
If $F_2$ is a mean-preserving spread of $F_1$, then
\[
H_{\gamma,F_2}(x)\ge H_{\gamma,F_1}(x)
\quad\text{for all }x,
\]
so the corresponding limit distribution of normalized top wages under $F_2$ is first-order stochastically smaller than under $F_1$.
Equivalently, for any increasing measurable $\varphi$ for which expectations are finite,
\[
\E\big[\varphi(Z_2)\big]\le \E\big[\varphi(Z_1)\big],
\qquad Z_k\sim H_{\gamma,F_k}.
\]
\end{corollary}

Corollary~\ref{cor:net_inequality} isolates a clean reduced-form message about the right tail.
Holding fixed the mean scale of opportunities and the common offer distribution, redistributing access more unequally lowers the distribution of maxima for a randomly drawn worker because the Laplace kernel is convex in types.
In this sense, homophily, segregation, or referral frictions matter here through the extent to which they reshape the opportunity distribution $F$.
At the same time, one should not overread the result at the furthest tail.
By Remark~\ref{rem:first_order_tail_equivalence}, the mean one normalization implies $1-H_{\gamma,F}(x)\sim v_\gamma(x)\sim 1-H_\gamma(x)$ whenever $v_\gamma(x)\downarrow 0$, so greater network inequality changes distributional levels and high quantiles, but not the first-order ultra-tail mass.
Translating that comparative static into welfare, equilibrium efficiency, or group incidence would require additional structure on search behavior, offer-quality heterogeneity, and market clearing.

\paragraph{A concrete rewiring counterfactual.}
Consider a referral-platform redesign, mentoring intervention, or local rewiring experiment that changes the empirical distribution of normalized degree from $F^0$ to $F^1$ while preserving the population mean.
If $F^1$ is a mean-preserving contraction of $F^0$, Corollary~\ref{cor:net_inequality} implies an everywhere improvement in the normalized top wage distribution for a randomly drawn worker.
If the change is not ordered by convex order, the adapted geometry below still provide a disciplined whole law comparison.

\paragraph{A simple numerical illustration.}
Consider a baseline network distribution
\[
F_0=0.8\,\delta_{0.5}+0.2\,\delta_{3},
\]
which can be read as a labor market with many peripheral workers and a smaller highly connected group, while preserving the normalization $\E[X]=1$.
At the centered threshold $x=0$, one has $v_\gamma(0)=1$ for every $\gamma$, so
\[
H_{\gamma,F_0}(0)=P_0(1)=0.8e^{-0.5}+0.2e^{-3}\approx 0.495,
\qquad
H_\gamma(0)=e^{-1}\approx 0.368.
\]
Thus the unequal-opportunity economy places about $12.7$ percentage points more mass below the centered threshold, equivalently about $12.7$ percentage points less mass above it, than the homogeneous benchmark.
The corresponding misallocation indices are $\mathcal M_1(F_0)=\E[|X-1|]=0.8$ and $\mathcal M_2(F_0)=\sqrt{\Var(X)}=1$.
This simple example makes concrete how a small, highly connected minority can coexist with a materially weaker distribution of normalized top wages for a randomly drawn worker.

\paragraph{Geometric counterfactual paths.}
When two economies with type distributions $F^0$ and $F^1$ are not ordered by convex order, the adapted-Wasserstein geometry can still provide, quite remarkably, a disciplined comparison.
Corollary~\ref{cor:adapted_geodesic_stability} constructs a canonical interpolation path $(F^t)_{t\in[0,1]}$ in the ambient transformed type geometry such that the induced law of top wages moves at most linearly in path length.
Corollary~\ref{cor:quantile_schedule_stability} turns the same result into an economically direct statement for the entire counterfactual schedule of top wage quantiles:
\[
\|Q_{\gamma,F^s}-Q_{\gamma,F^t}\|_{L^p(0,1)}
\le
C_{\gamma,p}\,|t-s|\,d_{\gamma,p}(F^0,F^1).
\]
Thus the adapted metric is not only a pointwise cdf device. It controls the whole wage distribution in quantile space, which is the natural object for many counterfactual exercises.
Because transformed coordinate geodesics need not preserve $\E[X]=1$ at intermediate times when $s_\gamma$ is nonlinear, this path should be read as a sensitivity device in ambient shape space, not as a literal mean-preserving policy path.
If a counterfactual must keep aggregate opportunity scale fixed throughout, Proposition~\ref{prop:renormalization_bridge} quantifies the discrepancy between the ambient path and its mean one renormalization, while the raw space mean-preserving geodesics of Section~\ref{sec:geometry} remain available when exact preservation is required throughout.

\paragraph{Robust tail prediction under measurement error in $F$.}
In practice, $F$ can be estimated from network data or inferred from proxies for the intensity of job contact.
The Wasserstein stability results in Section~\ref{sec:Wstab} then provide a direct way to translate estimation error in $F$ into error bounds for predicted tail outcomes.
The main linear bound is stated in the adapted metric $d_{\gamma,p}$.
When the available statistical control is instead in raw space Wasserstein distance, Proposition~\ref{prop:type_metric_controls} provides the additional bridge, which can be weaker in the Fr\'echet regime.

\begin{corollary}[Robustness of predicted tail distributions]\label{cor:net_robustness}
Let $\widehat F$ be an estimator of $F$ with $\E[X]=1$.
Assume the conditions of Theorem~\ref{thm:stability_HEV} hold for some $p\ge 1$.
Then the induced error on the predicted distribution of normalized top wages satisfies
\[
W_p\!\left(H_{\gamma,\widehat F},H_{\gamma,F}\right)
\le
C_{\gamma,p}\,d_{\gamma,p}\!\left(\widehat F,F\right),
\]
for the modulus $C_{\gamma,p}$ characterized in Section~\ref{sec:Wstab}.
Under the support and tail-shape conditions of Proposition~\ref{prop:type_metric_controls}, the adapted distance on the right-hand side can in turn be controlled by a standard Wasserstein distance on the type space.
Moreover, if $\psi$ is Lipschitz on the support of $H_{\gamma,F}$ and the relevant moments exist, then
\[
\left|\E[\psi(Z_{\widehat F})]-\E[\psi(Z_{F})]\right|
\le
\mathrm{Lip}(\psi)\,W_1\!\left(H_{\gamma,\widehat F},H_{\gamma,F}\right),
\]
where $Z_{\widehat F}\sim H_{\gamma,\widehat F}$ and $Z_F\sim H_{\gamma,F}$.
\end{corollary}

Corollary~\ref{cor:net_robustness} is the operational bridge between statistical or measurement uncertainty in network heterogeneity and uncertainty in predicted tail-wage outcomes, but only after the maintained comparison metric has been specified clearly.
If the researcher controls $d_{\gamma,p}$ directly, the bound is linear.
If the researcher begins from raw space Wasserstein error, the relevant conclusion is obtained only after applying Proposition~\ref{prop:type_metric_controls}, and in Fr\'echet settings that bridge can be merely H\"older.
For peripheral or intermittently inactive workers, the relevant issue is behavior near $X=0$.
Proposition~\ref{prop:domain_dgamma} makes the boundary explicit: the adapted transport bounds remain available as long as the transformed moments required by $d_{\gamma,p}$ are finite.
When those conditions fail, one should fall back on the pointwise or finite horizon comparisons rather than on the adapted metric.

\paragraph{Group heterogeneity, segregation, and tail gaps.}
The mapping can of course be aplied group by group.
If two groups $g\in\{A,B\}$ face different induced opportunity distributions $F_g$ due to segregation, homophily, or differential access to referrals, then the model implies group-specific limit laws $H_{\gamma,F_g}$.
The geometric tools in Sections~\ref{sec:geometry} and \ref{sec:Wstab} then provide two complementary comparisons:
order comparisons when $F_A$ dominates $F_B$ in convex order, and metric comparisons when $F_A$ and $F_B$ are close in Wasserstein distance but not ordered.

\subsection{Normative implications: opportunity reallocation and policy design}\label{subsec:normative}

Many interventions that affect labor market outcomes operate through changing access to opportunities, rather than directly changing offer values.
Examples include referral programs, mentoring and placement initiatives, changes in search subsidies, and institutional designs that affect who meets whom.
At an abstract level, such policies can all be represented as perturbations of the induced type distribution $F$.

Section~\ref{sec:entropic} provides a tractable design problem in which a planner chooses $F$ to optimize a tail-sensitive objective while controlling deviations from a baseline $F_0$ via relative entropy.
In the present setting, a canonical objective is expected utility of normalized top wages,
\[
U(F)=\E[u(Z)],\qquad Z\sim H_{\gamma,F},
\]
or a smoothed tail criterion that places higher marginal weight on large wage realizations.
The planner problem takes the form
\[
\sup_{F:\ \E_F[X]=1}\Big\{U(F)-\lambda D_{\mathrm{KL}}(F\|F_0)\Big\},
\]
where $\lambda$ governs the marginal cost of reshaping the cross-sectional distribution of access.
When $U$ is linear in $F$, including the expected-utility class identified via the stochastic representation in Section~\ref{subsec:frechet_linearization}, Theorem~\ref{thm:entropic_duality} yields a unique optimizer $F^\star$ given by an exponential tilt of $F_0$.

\paragraph{A discrete illustration with genuine degrees of freedom.}
To make the mean-one constraint non-vacuous, consider a baseline with three support points,
\[
F_0=\pi_1^0\,\delta_{x_1}+\pi_2^0\,\delta_{x_2}+\pi_3^0\,\delta_{x_3},
\qquad 0<x_1<x_2<x_3,
\]
with $\sum_j \pi_j^0=1$ and $\sum_j \pi_j^0 x_j=1$.
The optimal policy keeps the same support and reweights the masses according to
\[
\pi_j^\star
\propto
\pi_j^0\exp\!\left(\frac{\psi(x_j)+\eta^\star x_j}{\lambda}\right),
\qquad j=1,2,3,
\]
where $\eta^\star$ is chosen so that $\sum_j \pi_j^\star x_j=1$.
Pairwise odds satisfy
\[
\frac{\pi_i^\star/\pi_j^\star}{\pi_i^0/\pi_j^0}
=
\exp\!\left(\frac{\psi(x_i)-\psi(x_j)+\eta^\star(x_i-x_j)}{\lambda}\right).
\]
Unlike the two-point case, the mean-one restriction doesn't pin down the feasible distribution uniquely, so the entropy penalty and the score produce a genuine tradeoff.
This pairwise-odds equation is the discrete reduced-form analogue of the full optimal policy tradeoff in our network application: the tail score $\psi$, the implementation friction $\lambda$, and the shadow value $\eta^\star$ jointly determine how mass is reallocated across access types.

This upshot formalizes the idea that improving tail outcomes can require reallocating access across the population, but such reallocations are constrained by implementation frictions.
It also complements richer normative analyses of network policy by isolating the tail-sensitive, reduced-form margin through which access is reweighted.

\subsection{Microfoundations and discussion}\label{subsec:microfoundations}

The mixed-Poisson arrival structure can be grounded in standard network-based mechanisms of job search.
For instance, if workers receive job leads from contacts and each contact generates opportunities at approximately independent random times, then conditional on network exposure the number of leads over a horizon is naturally approximated by a Poisson random variable with mean proportional to a network index.
This is the reduced-form counterpart of the job search mechanisms discussed in the network-of-contacts literature overviewed in Section~\ref{subsec:literature}.

In modern settings, $X$ can be interpreted more broadly than degree.
For instance, it can incorporate platform-mediated matching intensity, differential search effort, or institution-specific access, while still entering the analysis only through its distribution $F$ and the normalization $\E[X]=1$.
This interpretation is especially natural when referral processes, homophily, or segregation generate persistent differences in access to opportunities across workers or groups.

\subsection{Empirical interface and identification caveats}\label{subsec:empirical}

The framework suggests a useful empirical interface, but the scope of identification should be delimited carefully.

\paragraph{Counts point-identify $F$ in population, but inversion is ill-posed.}
If one observes offer counts at a single horizon $\theta$, then the data identify the count distribution and its pgf.
Under the mixed-Poisson structure, that data-identified pgf pins down the Laplace transform $P_0$ on the interval $z\in[0,\theta]$.
Because $P_0$ is the Laplace transform of a positive measure, its values on any open interval determine the entire transform and therefore the mixing law $F$.
Population point identification is thus not the issue.
Rather, the difficulty is statistical: recovering $F$ from one observed interval of the transform is a severely ill-posed inverse problem, so multiple horizons, repeated observations, or parametric/semi-parametric structure remain valuable for stable estimation and regularization.

\paragraph{Network proxies require a measurement model.}
Observed degree, weighted degree, referral exposure, or centrality can be informative proxies for opportunity intensity, but they do not mechanically equal the latent type $X$.
Using them to estimate $F$ requires a maintained mapping from observed network statistics to the shifter of arrival rates entering the mixed-Poisson specification.
For the same reason, network proxies alone do not identify the adapted metric $d_{\gamma,p}$ without a maintained choice of tail index $\gamma$ and a maintained measurement model for how the observed proxy maps into the latent type entering the extreme-value law.

\paragraph{Extreme-based inversion requires a first-step tail analysis.}
Inferring $F$ from extreme outcomes requires the tail-limit ingredients, namely the index $\gamma$, the normalization $(a_\theta,b_\theta)$, and, for second-order refinements, the corresponding second-order objects, either to be known or to be estimated separately.
The full parent law $G$ need not be known for the EVT approximation itself.
The asymptotic expansion in Theorem~\ref{thm:second_order_expansion} useful precisely because it separates second-order EVT error from the heterogeneity kernel.
But it does not, on its own, solve the statistical recovery problem for $F$.

\paragraph{Error propagation is conditional on the sampling scheme and can be H\"older in Fr\'echet settings.}
Once an empirical approximation $\widehat F$ is available, Corollary~\ref{cor:net_robustness} converts metric error in $\widehat F$ into error in predicted tail laws.
If $0<\gamma<1$ and the assumptions of Theorem~\ref{thm:stability_HEV} hold with $p\gamma<1$, Proposition~\ref{prop:type_metric_controls} yields
\[
W_p\!\left(H_{\gamma,\widehat F},H_{\gamma,F}\right)
\le
C_{\gamma,p}\,W_p(\widehat F,F)^\gamma.
\]
Combined with Corollary~\ref{cor:quantile_schedule_stability}, the same argument gives
\[
\|Q_{\gamma,\widehat F}-Q_{\gamma,F}\|_{L^p(0,1)}
\le
C_{\gamma,p}\,W_p(\widehat F,F)^\gamma.
\]
Thus the bridge from raw space estimation error in $F$ to error in predicted extremes is generally only H\"older, not linear, in the economically important Fr\'echet regime.
If $W_p(\widehat F,F)=O_P(r_n)$, the induced law error is only $O_P(r_n^\gamma)$.
Standard benchmark Wasserstein rates such as \citet{FournierGuillin2015} therefore remain informative, but only after this nonlinear penalty is made explicit.
Shared-network data, however, feature cross-sectional dependence, so those rates should be treated as benchmarks rather than as automatic guarantees.
Additional graph-dependence assumptions would be needed for a full statistical theory.

\paragraph{Finite-horizon and asymptotic claims should be kept distinct.}
For prelimit predictions at a fixed horizon $\theta$, Proposition~\ref{prop:finite_theta_stability} is the relevant robustness statement.
For asymptotic tail approximations, Theorem~\ref{thm:stability_HEV} and Theorem~\ref{thm:second_order_expansion} are the relevant tools.
Keeping those objects separate helps prevent finite-sample robustness, asymptotic approximation, and statistical identification from being conflated.

\section{Discussion and conclusion}\label{sec:concl}

This paper studies extremes in environments where the number of opportunities is heterogeneous and random.
Starting from the mixed-Poisson heterogeneous agent extreme value limits in \citet{Mangin2025}, we treat the heterogeneity distribution as the primitive object and analyze the induced map $F\mapsto H_{\gamma,F}$ with order, Wasserstein-geometric, and complementary entropy-projection tools.

\subsection{Main takeaways}\label{subsec:summary}

Four points structure the contribution.

First, heterogeneous extremes admit a compact Laplace-transform representation.
Once the primitive offer distribution is in a classical domain of attraction, the effect of unequal access to opportunities is summarized by the mean-one distribution $F$ of draw intensities and its Laplace transform.
This makes heterogeneity analytically tractable without collapsing it to a scalar index.

Second, the paper's main new quantitative theorem is geometric.
Theorem~\ref{thm:stability_HEV} shows that, once the canonical coupling representation is in hand, perturbations in $F$ propagate Lipschitzly into Wasserstein perturbations of the entire induced law of extremes.
Corollary~\ref{cor:adapted_geodesic_stability} packages the same structure into canonical ambient interpolation paths, while Theorem~\ref{thm:second_order_expansion} cleanly separates second-order EVT approximation error from the heterogeneity kernel.
By contrast, convex-order comparisons and several pointwise inequalities are consequences of the Laplace form and are included because they make the operator economically usable.

Third, the labor market network application shows how these ingredients fit together in a familiar environment, but only in reduced form.
Network position maps into opportunity intensity, network inequality maps into tail distortions for a randomly drawn worker's top wage, and the adapted geometry controls the whole counterfactual quantile schedule together with the renormalization error induced by returning to the mean one slice.
Turning those objects into equilibrium or welfare claims would require additional structure beyond the scope of the present paper.

Fourth, the same type space supports a complementary design problem.
Under KL regularization and a mean constraint, the planner's problem is a one-marginal entropy projection with a unique exponential-tilt solution.
The connection to the positive analysis is structural rather than geometric: the same Laplace kernel governs both the heterogeneous extreme-value operator and the score functions of a broad class of linear objectives, including cdf-based criteria and expected utility of normalized extremes.

\subsection{Extensions}\label{subsec:extensions}

Several extensions seem natural.

\begin{itemize}
\item \textbf{Endogenous networks and strategic exposure.}
In applications where network position is chosen or co-determined with outcomes, one could embed the mixed-Poisson intensity into an equilibrium model of link formation or effort choice.
The present stability and design results can then be used either as reduced-form comparative statics or as primitives inside a fixed-point argument.

\item \textbf{Dynamic designs and Schr\"odinger bridges.}
Our entropy-regularized formulation in Section~\ref{sec:entropic} is static.
A dynamic extension would allow the planner to steer heterogeneity over time subject to intertemporal entropic costs, which naturally leads to Schr\"odinger-bridge-type problems.
That extension would be genuinely richer than the one-marginal projection studied here.

\item \textbf{Dynamic record processes and sequential search.}
Because each new offer can be viewed as a potential record, one natural extension is to study the full record process generated by heterogeneous arrival rates rather than only the terminal maximum.
Classical references in this context include \citet{ArnoldBalakrishnanNagaraja1998}.
Such an extension would connect the present geometry of extremes to questions about the timing of record improvements and threshold-crossing times.

\item \textbf{Beyond mixed Poisson counts and beyond independence.}
The mixed-Poisson structure provides tractability through a Laplace transform.
Other count models or dependence structures can be accommodated whenever one can represent the distribution of maxima by a tractable transform and obtain couplings that separate heterogeneity from idiosyncratic shocks.
The classical random-sample-size maxima literature provides the natural benchmark for that extension; see \citet{BarndorffNielsen1964}, \citet{Galambos1973}, or \citet{SilvestrovTeugels1998}.

\item \textbf{Statistical identification and inference.}
The rate expansion in Section~\ref{subsec:rates} and the stability results in Section~\ref{sec:Wstab} suggest a path toward inference for $F$ from tail observations or count data.
A full treatment would have to combine second-order tail estimation for $G$, stable recovery or regularized inversion of $F$ from observed count distributions or extremes, and a dependence-aware theory for network data.
The present paper isolates the operator-theoretic ingredients needed for that future step rather than claiming to complete it.
\end{itemize}

\paragraph{Conclusion.}
Extremes are economically salient precisely because they aggregate rare opportunities and heterogeneity in access.
The framework developed here provides a tractable way to quantify how heterogeneity shapes agent-level extremes, how robust those implications are to perturbations in heterogeneity, and how a planner can reshape heterogeneity when the objective is tail sensitive and policy changes are entropy penalized.
The transport results should be read on their natural geometric domain, Proposition~\ref{prop:renormalization_bridge} as the explicit link back to the mean-one economic slice, and the labor-network application as a disciplined reduced-form rather than a complete equilibrium model.

\appendix
\section{Proofs}\label{app:proofs}

\subsection{Auxiliary limit for maxima}\label{app:aux_evt}

We repeatedly use the elementary implication that turns convergence of powers into a first-order tail approximation.

\begin{lemma}[From powers to first-order tails]\label{lem:powers_to_tails}
Let $\{u_n\}_{n\ge 1}\subset[0,1)$ and suppose there exists $v\in[0,\infty)$ such that
\[
(1-u_n)^n \to e^{-v}\qquad\text{as }n\to\infty.
\]
Then $u_n\to 0$ and $n u_n \to v$.
\end{lemma}

\begin{proof}
Since $(1-u_n)^n \to e^{-v}\in(0,1]$, we must have $1-u_n\to 1$, hence $u_n\to 0$.

Write $(1-u_n)^n=\exp\big(n\log(1-u_n)\big)$.
Taking logs yields $n\log(1-u_n)\to -v$.

For $u\in[0,1)$, the inequalities
\[
-\frac{u}{1-u}\le \log(1-u)\le -u
\]
hold.
Applying them with $u=u_n$ and multiplying by $-n$ gives
\[
\frac{n u_n}{1-u_n}\ge -n\log(1-u_n)\ge n u_n.
\]
Letting $n\to\infty$ and using $u_n\to 0$ and $-n\log(1-u_n)\to v$, we obtain
\[
\limsup_{n\to\infty} n u_n \le v
\quad\text{and}\quad
\liminf_{n\to\infty} n u_n \ge v,
\]
so $n u_n\to v$.
\end{proof}

\subsection{Proof of Lemma~\ref{lem:mp_pgf}}\label{app:proof_mixed_poisson_pgf}

\begin{proof}
Fix $y\in[0,1]$ and $\theta>0$.
By conditioning on $X$ and using the probability generating function of a Poisson random variable,
\[
\E\!\left[y^{N(\theta)}\mid X\right]
=
\exp\!\big(-\theta X(1-y)\big).
\]
Taking expectations over $X$ yields
\[
\E\!\left[y^{N(\theta)}\right]
=
\E\!\left[\exp\!\big(-\theta X(1-y)\big)\right]
=
P_0\!\left(\theta(1-y)\right),
\]
which is the claimed identity.
\end{proof}

\subsection{Proof of Proposition~\ref{prop:HEV}}\label{app:proof_HEV}

\begin{proof}
Fix $x\in\R$.
Let $n=\lfloor \theta\rfloor$ and recall that $a_\theta=a_n$ and $b_\theta=b_n$ by definition.
Set
\[
x_\theta:=b_\theta+a_\theta x=b_n+a_n x.
\]

\medskip\noindent\textit{Step 1: reduce to a Laplace argument.}
By \eqref{eq:cdf_Mtheta},
\[
\Pr\!\left(\frac{M_\theta-b_\theta}{a_\theta}\le x\right)
=
\Pr\!\left(M_\theta\le x_\theta\right)
=
P_0\!\left(\theta\big(1-G(x_\theta)\big)\right).
\]
Since $P_0$ is the Laplace transform of a nonnegative random variable, it is continuous on $[0,\infty)$.

\medskip\noindent\textit{Step 2: identify the limit of $\theta(1-G(x_\theta))$.}
Assumption~\ref{ass:doa} implies
\[
\Pr\!\left(\frac{M_n-b_n}{a_n}\le x\right)
=
\Pr(M_n\le b_n+a_n x)
=
G(b_n+a_n x)^n
\to
H_\gamma(x)=e^{-v_\gamma(x)}
\]
at all continuity points $x$ of $H_\gamma$.
Define $u_n(x):=1-G(b_n+a_n x)\in[0,1)$.
Then $G(b_n+a_n x)^n=(1-u_n(x))^n$.

If $H_\gamma(x)>0$, equivalently $v_\gamma(x)<\infty$, Lemma~\ref{lem:powers_to_tails} gives
\[
n u_n(x)\to v_\gamma(x)
\qquad\text{and}\qquad
u_n(x)\to 0.
\]
Hence
\[
\big|\theta u_n(x)-n u_n(x)\big|
\le |\theta-n|\,u_n(x)
\le u_n(x)\to 0,
\]
since $|\theta-n|<1$.
Therefore
\[
\theta\big(1-G(b_n+a_n x)\big)=\theta u_n(x)\to v_\gamma(x).
\]

If instead $H_\gamma(x)=0$, then $(1-u_n(x))^n\to 0$.
We claim that $n u_n(x)\to\infty$.
If not, there would exist a subsequence $n_k$ and a constant $M<\infty$ such that $n_k u_{n_k}(x)\le M$.
Then $u_{n_k}(x)\le M/n_k\to 0$.
Passing to a further subsequence if needed, we may assume $n_k u_{n_k}(x)\to L$ for some finite $L\in[0,M]$.
Since $\log(1-u)=-u+o(u)$ as $u\downarrow 0$, it follows that
\[
n_k\log(1-u_{n_k}(x))\to -L,
\]
and therefore
\[
(1-u_{n_k}(x))^{n_k}\to e^{-L}>0,
\]
contradicting $(1-u_n(x))^n\to 0$.
Thus $n u_n(x)\to\infty$.
Since $\theta\ge n$, it follows that
\[
\theta u_n(x)\ge n u_n(x)\to\infty.
\]

\medskip\noindent\textit{Step 3: pass to the limit through $P_0$.}
If $H_\gamma(x)>0$, Step 2 and continuity of $P_0$ yield
\[
P_0\!\left(\theta(1-G(x_\theta))\right)
\to
P_0\!\big(v_\gamma(x)\big)
=
H_{\gamma,F}(x).
\]

If $H_\gamma(x)=0$, then $\theta(1-G(x_\theta))\to\infty$ by Step 2.
Because $e^{-zX}\to 1_{\{X=0\}}$ as $z\to\infty$ and Assumption~\ref{ass:mixedPoisson} imposes $\Pr(X=0)=0$,
dominated convergence gives
\[
P_0(z)=\E[e^{-zX}]\to 0
\qquad\text{as }z\to\infty.
\]
Hence
\[
P_0\!\left(\theta(1-G(x_\theta))\right)\to 0 = H_{\gamma,F}(x).
\]
This proves the asserted convergence at all continuity points of $H_{\gamma,F}$.

\medskip\noindent\textit{Degenerate case.}
If $F=\delta_1$, then $X=1$ almost surely, so $P_0(z)=\E[e^{-zX}]=e^{-z}$ and hence
\[
H_{\gamma,F}(x)=P_0(v_\gamma(x))=e^{-v_\gamma(x)}=H_\gamma(x).
\]
\end{proof}

\subsection{Proofs for Section~\ref{sec:geometry}}\label{app:proof_geometry}

\begin{proof}[Proof of Proposition~\ref{prop:laplace_order}]
Let $X_1\sim F_1$ and $X_2\sim F_2$ with $\E[X_1]=\E[X_2]=1$.
If $F_2$ is a mean-preserving spread of $F_1$, then $X_2$ dominates $X_1$ in convex order.
Equivalently,
\[
\E[\varphi(X_2)]\ge \E[\varphi(X_1)]
\quad\text{for every convex }\varphi\text{ for which both expectations are finite}.
\]

Fix $z\ge 0$ and consider $\varphi_z(x):=e^{-z x}$.
Then $\varphi_z$ is convex on $[0,\infty)$ (since $\varphi_z''(x)=z^2 e^{-z x}\ge 0$) and bounded by $1$.
Hence,
\[
P_0^{(2)}(z)=\E[e^{-z X_2}]\ge \E[e^{-z X_1}]=P_0^{(1)}(z)
\quad\text{for all }z\ge 0.
\]
Composing with $z=v_\gamma(x)\ge 0$ yields
\[
H_{\gamma,F_2}(x)=P_0^{(2)}\!\big(v_\gamma(x)\big)\ge P_0^{(1)}\!\big(v_\gamma(x)\big)=H_{\gamma,F_1}(x)
\quad\text{for all }x.
\]
Therefore, the distribution function under $F_2$ is everywhere larger, which is equivalent to first-order stochastic dominance in the direction that extremes are smaller under $F_2$.
\end{proof}

\begin{proof}[Proof of Corollary~\ref{cor:hetero_vs_hom}]
Apply Proposition~\ref{prop:laplace_order} with $F_1=\delta_1$ and $F_2=F$.
Since every mean one distribution on $[0,\infty)$ dominates $\delta_1$ in convex order, one obtains
\[
P_0(z)\ge e^{-z}
\qquad\text{for all }z\ge 0,
\]
and therefore $H_{\gamma,F}(x)\ge H_\gamma(x)$ for every $x$ with $v_\gamma(x)<\infty$.
If $F\neq \delta_1$, then strict convexity of $x\mapsto e^{-zx}$ for every $z>0$ and Jensen's inequality imply strict inequality for all $z>0$.
\end{proof}

\begin{proof}[Proof of Proposition~\ref{prop:convexity_along_geodesic}]
By Lemma~\ref{lem:quantile_geodesic}, the canonical monotone geodesic $(\mu_t)_{t\in[0,1]}$ satisfies
\[
Q_{\mu_t}(u)=(1-t)Q_\mu(u)+tQ_\nu(u),\qquad u\in(0,1).
\]
Since $\varphi$ is convex,
\[
\varphi\big(Q_{\mu_t}(u)\big)
\le
(1-t)\varphi\big(Q_\mu(u)\big)+t\varphi\big(Q_\nu(u)\big)
\qquad\text{for every }u\in(0,1).
\]
Integrating over $u\in(0,1)$ and using the quantile representation of integrals gives
\[
\int \varphi\,d\mu_t
\le
(1-t)\int \varphi\,d\mu+t\int \varphi\,d\nu.
\]
Applying the same argument to any pair of times $s,t\in[0,1]$ in place of $0,1$ yields convexity of $t\mapsto \int \varphi\,d\mu_t$ on $[0,1]$.
\end{proof}

\begin{proof}[Proof of Proposition~\ref{prop:pushforward_rep}]
For the Fr\'echet case, let $Z_\gamma$ be independent of $X\sim F$ and satisfy
\[
\Pr(Z_\gamma\le z)=\exp(-z^{-1/\gamma}),\qquad z>0.
\]
Then, for $z>0$,
\[
\Pr(X^\gamma Z_\gamma\le z\mid X)
=
\Pr\!\left(Z_\gamma\le zX^{-\gamma}\mid X\right)
=
\exp\!\left(-X z^{-1/\gamma}\right).
\]
Taking expectations over $X$ gives
\[
\Pr(X^\gamma Z_\gamma\le z)=\E\!\left[e^{-X z^{-1/\gamma}}\right]=P_0(z^{-1/\gamma}),
\]
which is the cdf of the heterogeneous Fr\'echet limit.
If $Z_{\mathrm{GEV}}:=(X^\gamma Z_\gamma-1)/\gamma$, then
\[
\Pr(Z_{\mathrm{GEV}}\le x)=\Pr(X^\gamma Z_\gamma\le 1+\gamma x)=P_0\!\big((1+\gamma x)^{-1/\gamma}\big)=H_{\gamma,F}(x).
\]

For the Gumbel case, let $E\sim\mathrm{Exp}(1)$ be independent of $X\sim F$ and define $Z:=\log(X/E)$.
Then, for every $x\in\R$,
\[
\Pr(Z\le x\mid X)=\Pr(E\ge Xe^{-x}\mid X)=e^{-Xe^{-x}}.
\]
Taking expectations yields
\[
\Pr(Z\le x)=\E[e^{-Xe^{-x}}]=P_0(e^{-x})=H_{0,F}(x).
\]
\end{proof}

\begin{proof}[Proof of Lemma~\ref{lem:kernel_expectation_frechet}]
By Proposition~\ref{prop:pushforward_rep}, $Z=X^\gamma Z_\gamma$ with $Z_\gamma$ independent of $X$.
Hence, by iterated expectation,
\[
\E[\psi(Z)]
=
\E\!\left[\E\!\left[\psi(X^\gamma Z_\gamma)\mid X\right]\right]
=
\E[\kappa_\psi(X)]
=
\int_0^\infty \kappa_\psi(x)\,F(dx).
\]
\end{proof}

\begin{proof}[Proof of Proposition~\ref{prop:gateaux}]
For every $x$ with $v_\gamma(x)<\infty$,
\[
H_{\gamma,F_\varepsilon}(x)
=
\int e^{-v_\gamma(x)u}\,F_\varepsilon(du)
=
\int e^{-v_\gamma(x)u}\,F(du)
+\varepsilon\int e^{-v_\gamma(x)u}\,\nu(du).
\]
Differentiating with respect to $\varepsilon$ at $0$ yields the claim.
\end{proof}

\begin{proof}[Proof of Proposition~\ref{prop:convexity_H_along_geodesic}]
Fix $z\ge 0$ and apply Proposition~\ref{prop:convexity_along_geodesic} with the convex function $\varphi_z(x)=e^{-zx}$.
Then
\[
t\mapsto \int e^{-zx}\,F^t(dx)=P_0^{\,t}(z)
\]
is convex on $[0,1]$.
Setting $z=v_\gamma(x)$ gives the convexity of $t\mapsto H_{\gamma,F^t}(x)$.
\end{proof}

\begin{proof}[Proof of Proposition~\ref{prop:neg_moment_geodesic_convex}]
On $(0,\infty)$, the function $\varphi(x)=x^{-\rho}$ is convex for every $\rho>0$.
Since the geodesic is supported on $[a,\infty)$, the function is also bounded on the support of every $F^t$.
Applying Proposition~\ref{prop:convexity_along_geodesic} with this $\varphi$ yields convexity of
\[
t\mapsto \int x^{-\rho}\,F^t(dx).
\]
\end{proof}

\subsection{Proofs for Section~\ref{sec:Wstab}}\label{app:proof_Wstab}

\begin{proof}[Proof of Proposition~\ref{prop:canonical_rep}]
Let $X\sim F$ and $E\sim\mathrm{Exp}(1)$ be independent, and define $Z_{\gamma,F}=w_\gamma(E/X)$.
For any $x$ such that $v_\gamma(x)\in(0,\infty)$,
\[
\Pr(Z_{\gamma,F}\le x\mid X)
=
\Pr\!\left(w_\gamma(E/X)\le x\mid X\right)
=
\Pr\!\left(E/X\ge v_\gamma(x)\mid X\right),
\]
since $w_\gamma$ is decreasing.
Therefore,
\[
\Pr(Z_{\gamma,F}\le x\mid X)=\exp\!\big(-Xv_\gamma(x)\big).
\]
Taking expectations yields
\[
\Pr(Z_{\gamma,F}\le x)=\E\!\left[e^{-Xv_\gamma(x)}\right]=P_0\!\big(v_\gamma(x)\big)=H_{\gamma,F}(x).
\]
The explicit formulas follow from $w_\gamma(t)=(t^{-\gamma}-1)/\gamma$ for $\gamma\neq 0$ and $w_0(t)=-\log t$.
\end{proof}

\begin{proof}[Proof of Theorem~\ref{thm:stability_HEV}]
\medskip\noindent\textit{Case $\gamma\neq 0$.}
Let $\mu_i:=s_\gamma\#F_i$ for $i\in\{1,2\}$, so that $d_{\gamma,p}(F_1,F_2)=W_p(\mu_1,\mu_2)$.
Let $(U_1,U_2)$ be an otimal coupling of $(\mu_1,\mu_2)$, and let $E\sim\mathrm{Exp}(1)$ be independent of $(U_1,U_2)$.
Set $V:=E^{-\gamma}$ and define
\[
\phi(u,v):=\frac{uv-1}{\gamma}.
\]
If $X_i:=s_\gamma^{-1}(U_i)$, then $X_i\sim F_i$ and
\[
\phi(U_i,V)=\frac{X_i^\gamma E^{-\gamma}-1}{\gamma}=w_\gamma(E/X_i).
\]
By Proposition~\ref{prop:canonical_rep}, $\Law(\phi(U_i,V))=H_{\gamma,F_i}$.
For each fixed $v$, the map $u\mapsto \phi(u,v)$ is Lipschitz with constant $|v|/|\gamma|$.
Hence Lemma~\ref{lem:random_lipschitz} gives
\[
W_p\!\left(H_{\gamma,F_1},H_{\gamma,F_2}\right)
\le
\frac{\big(\E[|E^{-\gamma}|^p]\big)^{1/p}}{|\gamma|}
W_p(\mu_1,\mu_2)
=
\frac{\big(\E[|E^{-\gamma}|^p]\big)^{1/p}}{|\gamma|}
 d_{\gamma,p}(F_1,F_2).
\]
This proves the bound for $\gamma\neq 0$.

\medskip\noindent\textit{Case $\gamma=0$.}
Let $\mu_i:=(\log)\#F_i$ for $i\in\{1,2\}$, so that $d_{0,p}(F_1,F_2)=W_p(\mu_1,\mu_2)$.
Let $(U_1,U_2)$ be an optimal coupling of $(\mu_1,\mu_2)$, let $E\sim\mathrm{Exp}(1)$ be independent, and define
\[
\phi(u,v):=u-\log v.
\]
If $X_i:=e^{U_i}$, then $X_i\sim F_i$ and
\[
\phi(U_i,E)=\log X_i-\log E=w_0(E/X_i).
\]
By Proposition~\ref{prop:canonical_rep}, $\Law(\phi(U_i,E))=H_{0,F_i}$.
For each fixed $v>0$, the map $u\mapsto \phi(u,v)$ is $1$-Lipschitz, so Lemma~\ref{lem:random_lipschitz} yields
\[
W_p\!\left(H_{0,F_1},H_{0,F_2}\right)
\le W_p(\mu_1,\mu_2)
=d_{0,p}(F_1,F_2).
\]
\end{proof}

\begin{proof}[Proof of Corollary~\ref{cor:adapted_geodesic_stability}]
By construction,
\[
d_{\gamma,p}(F^s,F^t)=W_p(G^s,G^t).
\]
Since $(G^t)_{t\in[0,1]}$ is the canonical monotone constant-speed $W_p$ geodesic from $G^0$ to $G^1$,
\[
W_p(G^s,G^t)=|t-s|\,W_p(G^0,G^1)=|t-s|\,d_{\gamma,p}(F^0,F^1).
\]
Applying Theorem~\ref{thm:stability_HEV} to $F^s$ and $F^t$ yields the second claim.
\end{proof}

\begin{proof}[Proof of Proposition~\ref{prop:renormalization_bridge}]
Let $m:=m(F)$ and let $Y:=s_\gamma(X)$ for $X\sim F$.
If $\gamma\neq 0$, then
\[
s_\gamma\!\left(\frac{X}{m}\right)=\left(\frac{X}{m}\right)^\gamma=m^{-\gamma}X^\gamma=m^{-\gamma}Y.
\]
Hence $s_\gamma\#\widetilde F$ is the law of $m^{-\gamma}Y$.
Because quantiles scale linearly under multiplication by a positive constant,
\[
d_{\gamma,p}(F,\widetilde F)^p
=
W_p\!\left(\Law(Y),\Law(m^{-\gamma}Y)\right)^p
=
\int_0^1 \left|Q_Y(u)-m^{-\gamma}Q_Y(u)\right|^p\,du.
\]
Therefore
\[
d_{\gamma,p}(F,\widetilde F)
=
\left|1-m^{-\gamma}\right|\left(\int_0^1 |Q_Y(u)|^p\,du\right)^{1/p}
=
\left|1-m^{-\gamma}\right|\left(\E[|Y|^p]\right)^{1/p},
\]
which is exactly the stated formula for $\gamma\neq 0$.
If $\gamma=0$, then $s_0(X/m)=\log X-\log m$, so $s_0\#\widetilde F$ is the translate of $s_0\#F$ by $-\log m$.
Again by the quantile representation,
\[
d_{0,p}(F,\widetilde F)
=
W_p\!\left((\log)\#F,\tau_{-\log m}\#((\log)\#F)\right)
=
|\log m|,
\]
where $\tau_c(y):=y+c$.
The Wasserstein bound for the induced extreme laws is then immediate from Theorem~\ref{thm:stability_HEV}.
\end{proof}

\begin{proof}[Proof of Corollary~\ref{cor:quantile_schedule_stability}]
The equality with $W_p\!\left(H_{\gamma,F_1},H_{\gamma,F_2}\right)$ is the one-dimensional quantile representation of Wasserstein distance.
The first inequality is Theorem~\ref{thm:stability_HEV}.
The geodesic statement then follows from Corollary~\ref{cor:adapted_geodesic_stability}.
\end{proof}

\begin{proof}[Proof of Corollary~\ref{cor:pointwise_cdf_stability}]
Fix $x$ in theinterior of the support of $H_\gamma$ and define
\[
f_x(z):=e^{-v_\gamma(x)z},\qquad z\ge 0.
\]
Then $f_x$ is $v_\gamma(x)$-Lipschitz on $[0,\infty)$ because
\[
|f_x'(z)|=v_\gamma(x)e^{-v_\gamma(x)z}\le v_\gamma(x).
\]
Since
\[
H_{\gamma,F_i}(x)=\int f_x(z)\,F_i(dz),
\]
Corollary~\ref{cor:lipschitz_test} implies
\[
|H_{\gamma,F_1}(x)-H_{\gamma,F_2}(x)|\le v_\gamma(x)W_1(F_1,F_2).
\]
\end{proof}

\begin{proof}[Proof of Proposition~\ref{prop:domain_dgamma}]
By definition,
\[
d_{\gamma,p}(F_1,F_2)=W_p\!\left(F_1\circ s_\gamma^{-1},\,F_2\circ s_\gamma^{-1}\right).
\]
A Wasserstein distance on $\R$ is finite if and only if both marginals belong to $\cP_p(\R)$.
Hence $d_{\gamma,p}(F_1,F_2)<\infty$ if and only if
\[
\int |y|^p\,\big(F_i\circ s_\gamma^{-1}\big)(dy)<\infty,\qquad i\in\{1,2\}.
\]
By change of variables this is equivalent to
\[
\int |s_\gamma(x)|^p\,F_i(dx)<\infty,\qquad i\in\{1,2\}.
\]
The three regime-specific statements follow by substituting $s_\gamma(x)=x^\gamma$ for $\gamma\neq 0$ and $s_0(x)=\log x$.
\end{proof}

\begin{proof}[Proof of Proposition~\ref{prop:type_metric_controls}]
\medskip\noindent\textit{Case 1: $0<\gamma\le 1$.}
Let $(X,Y)$ be any coupling of $(F_1,F_2)$.
Since $x\mapsto x^\gamma$ is $\gamma$-H\"older on $[0,\infty)$,
\[
|X^\gamma-Y^\gamma|\le |X-Y|^\gamma.
\]
Therefore,
\[
\E\!\left[|X^\gamma-Y^\gamma|^p\right]
\le
\E\!\left[|X-Y|^{\gamma p}\right]
\le
\left(\E\!\left[|X-Y|^p\right]\right)^\gamma,
\]
where the last step is Jensen's inequality applied to the concave map $t\mapsto t^\gamma$.
Taking the infimum over couplings and then $p$th roots yields
\[
d_{\gamma,p}(F_1,F_2)\le W_p(F_1,F_2)^\gamma.
\]

\medskip\noindent\textit{Case 2: $\gamma<0$.}
On $[a,\infty)$, the derivative of $x\mapsto x^\gamma$ satisfies
\[
\left|\frac{d}{dx}x^\gamma\right|=|\gamma|x^{\gamma-1}\le |\gamma|a^{\gamma-1}.
\]
Hence $x\mapsto x^\gamma$ is $|\gamma|a^{\gamma-1}$-Lipschitz on $[a,\infty)$.
Lemma~\ref{lem:lipschitz_pushforward} gives
\[
d_{\gamma,p}(F_1,F_2)=W_p\!\left((x\mapsto x^\gamma)_\#F_1,(x\mapsto x^\gamma)_\#F_2\right)
\le |\gamma|a^{\gamma-1}W_p(F_1,F_2).
\]

\medskip\noindent\textit{Case 3: $\gamma=0$.}
On $[a,\infty)$, the derivative of $\log x$ is bounded by $a^{-1}$.
Thus $\log$ is $a^{-1}$-Lipschitz there, and Lemma~\ref{lem:lipschitz_pushforward} yields
\[
d_{0,p}(F_1,F_2)=W_p\!\left((\log)_\#F_1,(\log)_\#F_2\right)
\le a^{-1}W_p(F_1,F_2).
\]
\end{proof}

\begin{proof}[Proof of Proposition~\ref{prop:lip_functionals_limit}]
Apply Corollary~\ref{cor:lipschitz_test} with $\mu=H_{\gamma,F_1}$, $\nu=H_{\gamma,F_2}$, and the $L$-Lipschitz function $\psi$.
\end{proof}

\begin{proof}[Proof of Proposition~\ref{prop:finite_theta_stability}]
Fix $x\in\R$ and define
\[
f_x(z):=\exp\!\big(-\theta z(1-G(x))\big),\qquad z\ge 0.
\]
Then
\[
|f_x'(z)|=\theta(1-G(x))\exp\!\big(-\theta z(1-G(x))\big)\le \theta(1-G(x)),
\]
so $f_x$ is $\theta(1-G(x))$-Lipschitz on $[0,\infty)$.
Because
\[
\Pr_F(M_\theta\le x)=\int f_x(z)\,F(dz),
\]
Corollary~\ref{cor:lipschitz_test} gives
\[
\left|\Pr_{F_1}(M_\theta\le x)-\Pr_{F_2}(M_\theta\le x)\right|
\le
\theta(1-G(x))W_1(F_1,F_2).
\]
\end{proof}

\begin{proof}[Proof of Theorem~\ref{thm:second_order_expansion}]
Fix a compact set $K$ contained in the interior of the support of $H_\gamma$.
By Assumption~\ref{ass:second_order}, there exists a remainder $\delta_\theta(x)$ such that
\[
\theta\big(1-G(b(\theta)+a(\theta)x)\big)
=
 v_\gamma(x)+A(\theta)h_\gamma(x)+\delta_\theta(x),
\qquad x\in K,
\]
with
\[
\sup_{x\in K}\frac{|\delta_\theta(x)|}{|A(\theta)|}\to 0.
\]
Write
\[
\Delta_\theta(x):=A(\theta)h_\gamma(x)+\delta_\theta(x).
\]
Since $h_\gamma$ is locally bounded and $K$ is compact, there exists $B_K<\infty$ such that $|h_\gamma(x)|\le B_K$ on $K$.
Hence
\[
\sup_{x\in K}|\Delta_\theta(x)|=O(|A(\theta)|).
\]

By \eqref{eq:cdf_Mtheta},
\[
\Pr(Z_\theta\le x)=P_0\!\left(v_\gamma(x)+\Delta_\theta(x)\right),\qquad x\in K.
\]
Set
\[
m_K:=\inf_{x\in K} v_\gamma(x)>0,
\qquad
M_K:=\sup_{x\in K} v_\gamma(x)<\infty.
\]
For $t\ge m_K/2$, differentiation under the integral sign gives
\[
P_0'(t)=-\E\!\left[Xe^{-tX}\right],
\qquad
P_0''(t)=\E\!\left[X^2e^{-tX}\right].
\]
Moreover, for every $t>0$,
\[
X^2e^{-tX}\le \sup_{y\ge 0} y^2e^{-y}\,t^{-2}=4e^{-2}t^{-2},
\]
so
\[
\sup_{t\ge m_K/2} |P_0''(t)|\le \frac{16e^{-2}}{m_K^2}<\infty.
\]
Therefore $P_0$ has uniformly bounded second derivative on a neighborhood of $\{v_\gamma(x):x\in K\}$.
A second-order Taylor expansion yields, uniformly for $x\in K$,
\[
P_0\!\left(v_\gamma(x)+\Delta_\theta(x)\right)
=
P_0\!\left(v_\gamma(x)\right)+P_0'\!\left(v_\gamma(x)\right)\Delta_\theta(x)+R_\theta(x),
\]
where
\[
\sup_{x\in K}|R_\theta(x)|\le C_K\sup_{x\in K}|\Delta_\theta(x)|^2=O\!\big(A(\theta)^2\big)
=o\!\big(A(\theta)\big).
\]
Since $P_0(v_\gamma(x))=H_{\gamma,F}(x)$ and $P_0'(v_\gamma(x))=-\E[Xe^{-v_\gamma(x)X}]$, we obtain
\[
\Pr(Z_\theta\le x)
=
H_{\gamma,F}(x)-A(\theta)h_\gamma(x)\E\!\left[Xe^{-v_\gamma(x)X}\right]+r_{\theta,K}(x),
\]
where
\[
r_{\theta,K}(x):=-\delta_\theta(x)\E\!\left[Xe^{-v_\gamma(x)X}\right]+R_\theta(x).
\]
Because $Xe^{-v_\gamma(x)X}\le e^{-1}/m_K$ uniformly on $K$, the first term is $o(A(\theta))$ uniformly on $K$; the second term is already $o(A(\theta))$ uniformly on $K$.
Hence
\[
\sup_{x\in K}\frac{|r_{\theta,K}(x)|}{|A(\theta)|}\to 0,
\]
which proves the theorem.
\end{proof}

\subsection{Proofs for Section~\ref{sec:entropic}}\label{app:proof_entropic}

\begin{proof}[Proof of Theorem~\ref{thm:entropic_duality}]
We write the design problem \eqref{eq:design} in the linear case as
\[
\sup_{F:\ \E_F[X]=1}\left\{\int \psi(x)\,F(dx)-\lambda D_{\mathrm{KL}}(F\|F_0)\right\}.
\]
Throughout, we restrict attention to $F\ll F_0$, since otherwise $D_{\mathrm{KL}}(F\|F_0)=+\infty$.

\medskip\noindent\textit{Step 1: weak duality.}
Fix $\eta\in\mathcal D$.
For any feasible $F$ with $\E_F[X]=1$,
\[
\int \psi\,dF-\lambda D_{\mathrm{KL}}(F\|F_0)
=
\int (\psi+\eta x)\,dF-\lambda D_{\mathrm{KL}}(F\|F_0)-\eta.
\]
Applying Theorem~\ref{thm:DV} with $P=F_0$ and $f=(\psi+\eta x)/\lambda$ gives
\[
\int (\psi+\eta x)\,dF-\lambda D_{\mathrm{KL}}(F\|F_0)
\le
\lambda \log Z(\eta).
\]
Hence, for every feasible $F$ and every $\eta\in\mathcal D$,
\[
\int \psi\,dF-\lambda D_{\mathrm{KL}}(F\|F_0)
\le
\lambda \log Z(\eta)-\eta.
\]
Taking the supremum over feasible $F$ and the the infimum over $\eta\in\mathcal D$ yields the weak-duality bound
\[
\sup_{F:\E_F[X]=1}\Big\{\int \psi\,dF-\lambda D_{\mathrm{KL}}(F\|F_0)\Big\}
\le
\inf_{\eta\in\mathcal D}\left\{\lambda \log Z(\eta)-\eta\right\}.
\]

\medskip\noindent\textit{Step 2: the dual objective is smooth and has a unique minimizer.}
Define
\[
\mathcal J(\eta):=\lambda \log Z(\eta)-\eta,\qquad \eta\in\mathcal D.
\]
Let $I=[a,b]\subset\mathcal D$ be compact.
For every $\eta\in I$ and every $x>0$,
\[
\exp\!\left(\frac{\psi(x)+\eta x}{\lambda}\right)
\le
\exp\!\left(\frac{\psi(x)+b x}{\lambda}\right),
\]
and the latter is integrable by Assumption~\ref{ass:linear_objective}.
Similarly,
\[
x\exp\!\left(\frac{\psi(x)+\eta x}{\lambda}\right)
\le
x\exp\!\left(\frac{\psi(x)+b x}{\lambda}\right),
\]
and
\[
x^2\exp\!\left(\frac{\psi(x)+\eta x}{\lambda}\right)
\le
x^2\exp\!\left(\frac{\psi(x)+b x}{\lambda}\right),
\]
with both dominating functions integrable.
Therefore dominated convergence implies that $Z$ is twice continuously differentiable on $\mathcal D$, with
\[
Z'(\eta)=\frac{1}{\lambda}M_1(\eta),
\qquad
Z''(\eta)=\frac{1}{\lambda^2}M_2(\eta).
\]
Hence
\[
\mathcal J'(\eta)=\frac{M_1(\eta)}{Z(\eta)}-1=m(\eta)-1,
\]
and
\[
\mathcal J''(\eta)=\frac{1}{\lambda}\Var_{F_\eta}(X)\ge 0,
\]
where $F_\eta$ is the exponential tilt defined by
\[
\frac{dF_\eta}{dF_0}(x)=\frac{\exp\!\left(\frac{\psi(x)+\eta x}{\lambda}\right)}{Z(\eta)}.
\]
We claim that in fact $\Var_{F_\eta}(X)>0$ for every $\eta\in\mathcal D$.
Indeed, if $\Var_{F_\eta}(X)=0$, then $F_\eta$ would be a Dirac mass.
But the density $dF_\eta/dF_0$ is strictly positive on the support of $F_0$, so this can happen only if $F_0$ itself is a Dirac mass.
In that case $m(\eta)\equiv 1$, contradicting the assumption that there exist $\eta_-,\eta_+\in\mathcal D$ with $m(\eta_-)<1<m(\eta_+)$.
Therefore
\[
\mathcal J''(\eta)>0\qquad\text{for all }\eta\in\mathcal D,
\]
so $\mathcal J$ is strictly convex on $\mathcal D$ and $\mathcal J'$ is strictly increasing and continuous.
By Assumption~\ref{ass:linear_objective},
\[
\mathcal J'(\eta_-)=m(\eta_-)-1<0< m(\eta_+)-1=\mathcal J'(\eta_+).
\]
The intermediate value theorem therefore yields some $\eta^\star\in(\eta_-,\eta_+)$ such that $\mathcal J'(\eta^\star)=0$, that is, $m(\eta^\star)=1$.
Because $\mathcal J'$ is strictly increasing, this root is unique.
Since $\mathcal J$ is strictly convex, the same $\eta^\star$ is the unique minimizer of $\mathcal J$ on $\mathcal D$.

\medskip\noindent\textit{Step 3: construct the primal optimizer and prove strong duality.}
Let $F^\star:=F_{\eta^\star}$.
By construction,
\[
\E_{F^\star}[X]=m(\eta^\star)=1,
\]
so $F^\star$ is feasible.
Because Assumption~\ref{ass:linear_objective} also gives
\[
\int |\psi(x)|\exp\!\left(\frac{\psi(x)+\eta^\star x}{\lambda}\right)\,F_0(dx)<\infty,
\]
and because $\E_{F^\star}[X]=1$ implies
\[
\int x\exp\!\left(\frac{\psi(x)+\eta^\star x}{\lambda}\right)\,F_0(dx)=Z(\eta^\star)<\infty,
\]
one also has $\int |f|e^f\,dF_0<\infty$ for $f=(\psi+\eta^\star x)/\lambda$.
Therefore Theorem~\ref{thm:DV} attains equality at $F^\star$ for that function $f$.
Hence
\[
\int (\psi+\eta^\star x)\,dF^\star-\lambda D_{\mathrm{KL}}(F^\star\|F_0)
=
\lambda \log Z(\eta^\star).
\]
Using $\E_{F^\star}[X]=1$, we obtain
\[
\int \psi\,dF^\star-\lambda D_{\mathrm{KL}}(F^\star\|F_0)
=
\lambda \log Z(\eta^\star)-\eta^\star
=\mathcal J(\eta^\star).
\]
Because $\eta^\star$ minimizes $\mathcal J$, this value coincide with the weak-duality upper bound.
Therefore strong duality holds, $F^\star$ is optimal, and \eqref{eq:dual_eta} and \eqref{eq:tilt} follow.

\medskip\noindent\textit{Step 4: uniqueness of the primal optimizer.}
The map $F\mapsto \int \psi\,dF$ is linear, and $F\mapsto D_{\mathrm{KL}}(F\|F_0)$ is strictly convex on $\{F\ll F_0\}$.
Hence
\[
F\mapsto \int \psi\,dF-\lambda D_{\mathrm{KL}}(F\|F_0)
\]
is strictly concave on the convex feasible set $\{F:\E_F[X]=1,\ F\ll F_0\}$.
A strictly concave objective has at most one maximizer, so the optimal $F^\star$ is unique.
\end{proof}

\section{Technical tools}\label{app:tools}

\subsection{Wasserstein geometry and duality}\label{app:tools_wass}

We collect standard facts about Wasserstein distance that are used throughout the paper.
References include \citet{Villani2009} and \citet{Santambrogio2015}.

Let $(\mathsf X,d)$ be a Polish metric space.
For $p\ge 1$, write $\cP_p(\mathsf X)$ for the set of Borel probility measures on $\mathsf X$ with finite $p$th moment.

\begin{definition}[Wasserstein distance]\label{def:Wp}
For $\mu,\nu\in\cP_p(\mathsf X)$,
\[
W_p(\mu,\nu):=\left(\inf_{\pi\in\Gamma(\mu,\nu)} \int d(x,y)^p\,\pi(dx,dy)\right)^{1/p},
\]
where $\Gamma(\mu,\nu)$ denotes the set of couplings of $\mu$ and $\nu$.
\end{definition}

\begin{lemma}[Upper bound from any coupling]\label{lem:any_coupling_bound}
Let $p\ge 1$ and let $U,V$ be random variables in $\mathsf X$.
If $\Law(U)=\mu$ and $\Law(V)=\nu$, then
\[
W_p(\mu,\nu)\le \left(\E[d(U,V)^p]\right)^{1/p}.
\]
\end{lemma}

\begin{proof}
Let $\pi:=\Law(U,V)\in\Gamma(\mu,\nu)$.
By Definition~\ref{def:Wp},
\[
W_p^p(\mu,\nu)\le \int d(x,y)^p\,\pi(dx,dy)=\E[d(U,V)^p].
\]
\end{proof}

\begin{lemma}[One-dimensional quantile representation]\label{lem:quantile_rep}
Let $p\ge 1$ and let $\mu,\nu\in\cP_p(\R)$, where $\R$ is equipped with the metric $d(x,y)=|x-y|$.
Let $F_\mu,F_\nu$ denote the distribution functions and define the quantile functions
\[
Q_\mu(u):=\inf\{x\in\R: F_\mu(x)\ge u\},\qquad u\in(0,1),
\]
and similarly for $Q_\nu$.
Then
\[
W_p^p(\mu,\nu)=\int_0^1 |Q_\mu(u)-Q_\nu(u)|^p\,du.
\]
\end{lemma}

\begin{theorem}[Kantorovich-Rubinstein duality for $W_1$]\label{thm:KR}
Let $\mu,\nu\in\cP_1(\mathsf X)$.
Then
\[
W_1(\mu,\nu)=\sup\left\{\int f\,d\mu-\int f\,d\nu:\ f:\mathsf X\to\R,\ \|f\|_{\mathrm{Lip}}\le 1\right\},
\]
where $\|f\|_{\mathrm{Lip}}:=\sup_{x\neq y}\frac{|f(x)-f(y)|}{d(x,y)}$.
\end{theorem}

\begin{corollary}[Lipschitz test functions]\label{cor:lipschitz_test}
Let $\mu,\nu\in\cP_1(\mathsf X)$ and let $f:\mathsf X\to\R$ be $L$-Lipschitz.
Then
\[
\left|\int f\,d\mu-\int f\,d\nu\right|\le L\,W_1(\mu,\nu).
\]
\end{corollary}

\begin{proof}
If $L=0$, the claim is immediate.
Otherwise, $f/L$ has Lipschitz seminorm at most $1$, and Theorem~\ref{thm:KR} yields
\[
\left|\int f\,d\mu-\int f\,d\nu\right|
=
L\left|\int (f/L)\,d\mu-\int (f/L)\,d\nu\right|
\le
L\,W_1(\mu,\nu).
\]
\end{proof}

\begin{lemma}[Deterministic Lipschitz pushforward]\label{lem:lipschitz_pushforward}
Let $p\ge 1$, let $\mu,\nu\in\cP_p(\mathsf X)$, and let $g:\mathsf X\to\mathsf Y$ be $L$-Lipschitz between metric spaces $(\mathsf X,d_{\mathsf X})$ and $(\mathsf Y,d_{\mathsf Y})$.
Then the pushforward measures $g_\#\mu$ and $g_\#\nu$ belong to $\cP_p(\mathsf Y)$ and satisfy
\[
W_p(g_\#\mu,g_\#\nu)\le L\,W_p(\mu,\nu).
\]
\end{lemma}

\begin{proof}
Let $\pi\in\Gamma(\mu,\nu)$ and let $(U,V)\sim\pi$.
Then $(g(U),g(V))$ is a coupling of $(g_\#\mu,g_\#\nu)$ and
\[
\E\!\left[d_{\mathsf Y}(g(U),g(V))^p\right]
\le
L^p\,\E\!\left[d_{\mathsf X}(U,V)^p\right].
\]
Taking the infimum over $\pi\in\Gamma(\mu,\nu)$ and then $p$th roots yields the result.
\end{proof}

A central step in Section~\ref{sec:Wstab} uses a coupling where the mapping is Lischitz conditional on a random environment.
The following lemma isolates this idea.

\begin{lemma}[Random Lipschitz contraction]\label{lem:random_lipschitz}
Let $p\ge 1$.
Let $U_1,U_2$ be real-valued random variables with laws $\mu_1,\mu_2\in\cP_p(\R)$.
Let $V$ be an auxiliary random variable taking values in a measurable space $\mathsf V$, independent of $(U_1,U_2)$.
Let $\phi:\R\times\mathsf V\to\R$ be measurable and assume there exists a measurable function $L:\mathsf V\to[0,\infty)$ such that
\[
|\phi(u,v)-\phi(u',v)|\le L(v)\,|u-u'| \quad\text{for all }u,u'\in\R\text{ and all }v\in\mathsf V,
\]
and $\E[L(V)^p]<\infty$.
Then
\[
W_p\!\left(\Law(\phi(U_1,V)),\Law(\phi(U_2,V))\right)\le \left(\E[L(V)^p]\right)^{1/p} W_p(\mu_1,\mu_2).
\]
\end{lemma}

\begin{proof}
Let $(\widetilde U_1,\widetilde U_2)$ be an optimal coupling of $\mu_1$ and $\mu_2$.
Enlarge the probability space so tat $V$ is independent of $(\widetilde U_1,\widetilde U_2)$ and has the same law as in the statement.
Define $\widetilde Z_i:=\phi(\widetilde U_i,V)$ for $i\in\{1,2\}$, so that $\Law(\widetilde Z_i)=\Law(\phi(U_i,V))$.
By Lemma~\ref{lem:any_coupling_bound} and the Lipschitz property of $\phi$,
\[
W_p^p\!\left(\Law(\phi(U_1,V)),\Law(\phi(U_2,V))\right)
\le \E\!\left[|\widetilde Z_1-\widetilde Z_2|^p\right]
\le \E\!\left[L(V)^p\right]\E\!\left[|\widetilde U_1-\widetilde U_2|^p\right].
\]
Since $(\widetilde U_1,\widetilde U_2)$ is optimal, $\E[|\widetilde U_1-\widetilde U_2|^p]=W_p^p(\mu_1,\mu_2)$, and taking $p$th roots yields the claim.
\end{proof}

\begin{remark}
Lemma~\ref{lem:random_lipschitz} is applied in Section~\ref{sec:Wstab} with random transformations induced by the canonical coupling representation in Section~\ref{subsec:canonical_coupling}.
In particular, it is used with maps of the form $\phi(u,v)=\frac{uv-1}{\gamma}$ when $\gamma\neq 0$ and $\phi(u,v)=u-\log v$ when $\gamma=0$.
\end{remark}

\subsection{Entropic variational identities}\label{app:tools_entropic}

Let $P$ be a probability measure on a measurable space $(\mathsf X,\mathcal F)$.
The basic variational identity below is the Donsker-Varadhan formula; see \citet{DonskerVaradhan1975}.
For $Q\ll P$, define the relative entropy
\[
D(Q\|P):=\int \log\left(\frac{dQ}{dP}\right)\,dQ,
\]
and set $D(Q\|P)=+\infty$ if $Q$ is not absolutely continuous with respect to $P$.

\begin{theorem}[Donsker-Varadhan variational formula]\label{thm:DV}
Let $f:\mathsf X\to\R$ be measurable and assume $\int e^{f}\,dP<\infty$.
Then
\[
\log\int e^{f}\,dP
=
\sup_{Q\in\cP(\mathsf X)}\left\{\int f\,dQ - D(Q\|P)\right\}.
\]
Define the Gibbs tilt $Q^\star$ by
\[
\frac{dQ^\star}{dP}(x)=\frac{e^{f(x)}}{\int e^{f}\,dP}.
\]
If, in addition, $\int |f|e^{f}\,dP<\infty$ (equivalently, $\int |f|\,dQ^\star<\infty$ and $D(Q^\star\|P)<\infty$), then the supremum is attained at $Q^\star$.
\end{theorem}

\begin{corollary}[Entropy penalization in minimization form]\label{cor:DV_min}
Let $g:\mathsf X\to\R$ be measurable and assume $\int e^{-g}\,dP<\infty$.
Then
\[
-\log\int e^{-g}\,dP
=
\inf_{Q\in\cP(\mathsf X)}\left\{\int g\,dQ + D(Q\|P)\right\}.
\]
Define the Gibbs tilt $Q^\star$ by
\[
\frac{dQ^\star}{dP}(x)=\frac{e^{-g(x)}}{\int e^{-g}\,dP}.
\]
If, in addition, $\int |g|e^{-g}\,dP<\infty$, then the infimum is attained at $Q^\star$.
\end{corollary}

\begin{lemma}[Entropy penalization with linear constraints]\label{lem:entropy_linear_constraints}
Let $P\in\cP(\mathsf X)$ and let $g_0,g_1,\dots,g_m:\mathsf X\to\R$ be measurable.
Fix $\varepsilon>0$ and constants $c_1,\dots,c_m\in\R$.
Assume the feasible set
\[
\mathcal Q:=\left\{Q\in\cP(\mathsf X): \int g_k\,dQ=c_k,\ k=1,\dots,m\right\}
\]
is nonempty and contains some $Q_0\ll P$ with $D(Q_0\|P)<\infty$ and $\int |g_0|\,dQ_0<\infty$.
Let
\[
\mathcal C:=\left\{\left(\int g_1\,dQ,\dots,\int g_m\,dQ\right): Q\in\cP(\mathsf X),\ Q\ll P,\ D(Q\|P)<\infty,\ \int |g_0|\,dQ<\infty\right\}
\subseteq \R^m.
\]
Assume the target moment vector $c:=(c_1,\dots,c_m)$ belongs to the relative interior $\operatorname{ri}(\mathcal C)$.
Assume also that the value of the problem
\[
\inf_{Q\in\mathcal Q}
\left\{
\int g_0\,dQ + \varepsilon D(Q\|P)
\right\}
\]
is finite and attained at some $Q^\star$.
Then there exist finite Lagrange multipliers $\theta_1,\dots,\theta_m\in\R$ such that $Q^\star$ admits the exponential-tilt representation
\[
\frac{dQ^\star}{dP}(x)=
\frac{\exp\!\left(-\frac{1}{\varepsilon}\left(g_0(x)+\sum_{k=1}^m \theta_k g_k(x)\right)\right)}
{\int \exp\!\left(-\frac{1}{\varepsilon}\left(g_0+\sum_{k=1}^m \theta_k g_k\right)\right)\,dP}.
\]
\end{lemma}

\begin{remark}
Lemma~\ref{lem:entropy_linear_constraints} is in fact a general constrained-entropy projection result.
Section~\ref{sec:entropic} uses a sharper one-dimensional argument tailored to the mean constraint $\E_F[X]=1$ so that the existence of a finite dual optimizer can be stated directly in terms of the tilted mean.
\end{remark}

\subsection{Background on entropic optimal transport}\label{app:tools_eot}

Let $\mu\in\cP(\mathsf X)$ and $\nu\in\cP(\mathsf Y)$ be Borel probability measures on Polish spaces and let $c:\mathsf X\times\mathsf Y\to\R\cup\{+\infty\}$ be a measurable cost.
Let $\Gamma(\mu,\nu)$ denote the set of couplings with marginals $(\mu,\nu)$.
For $\varepsilon>0$, define the entropic optimal transport problem
\[
\inf_{\pi\in\Gamma(\mu,\nu)}\left\{\int c\,d\pi + \varepsilon D(\pi\|\mu\otimes\nu)\right\}.
\]
We use only structural consequences
The relevant ingredients are the duality and the Gibbs form of optimizers.
See, e.g., \citet{Leonard2014}, \citet{PeyreCuturi2019}, \citet{GhosalNutzBernton2022}, and \citet{Nutz2022}; and see also \citet{Cuturi2013} for the computational Sinkhorn formulation that made entropy regularization central in applied optimal transport.

\begin{theorem}[Entropic optimal transport duality and Gibbs form]\label{thm:eot_duality}
Assume the primal problem is finite for some $\varepsilon>0$ and that an optimizer exists.
Then there exist measurable functions $\varphi:\mathsf X\to\R$ and $\psi:\mathsf Y\to\R$ such that an optimal coupling $\pi^\star$ admits a density with respect to $\mu\otimes\nu$ of the form
\[
\frac{d\pi^\star}{d(\mu\otimes\nu)}(x,y)
=
\exp\!\left(\frac{\varphi(x)+\psi(y)-c(x,y)}{\varepsilon}\right),
\]
with $(\varphi,\psi)$ chosen so that $\pi^\star\in\Gamma(\mu,\nu)$.
Moreover, $(\varphi,\psi)$ solve a concave dual maximization problem obtained by Fenchel duality.
\end{theorem}

\begin{remark}
Note that Theorem~\ref{thm:eot_duality} is stated at a level sufficient for our purposes.
Section~\ref{sec:entropic} uses the same convex-duality logic in a reduced setting where the penalty is relative entropy with respect to a baseline $F_0$ and the decision variable is the marginal distribution $F$.
In that sense, the planner problem can be read as a one-marginal entropy projection derived from the broader Schr\"odinger and entropic-transport framework.
\end{remark}

\begingroup
\footnotesize
\linespread{0.96}\selectfont
\bibliographystyle{ecta}
\bibliography{references}

@article{Mangin2025,
  author  = {Mangin, Sephorah},
  title   = {Extreme Value Theory with Heterogeneous Agents},
  journal = {Econometrica},
  year    = {2026},
  note    = {Forthcoming}
}

@article{BobbiaDombryVarron2019,
  author  = {Bobbia, Benjamin and Dombry, Cl{\'e}ment and Varron, Davit},
  title   = {The Coupling Method in Extreme Value Theory},
  journal = {Bernoulli},
  year    = {2021},
  volume  = {27},
  number  = {3},
  pages   = {1824--1850},
  doi     = {10.3150/20-BEJ1293}
}

@misc{Nutz2022,
  author       = {Nutz, Marcel},
  title        = {Introduction to Entropic Optimal Transport},
  howpublished = {Lecture notes},
  year         = {2022},
  note         = {Available online}
}

@article{Leonard2014,
  author  = {L{\'e}onard, Christian},
  title   = {A Survey of the Schr{\"o}dinger Problem and Some of Its Connections with Optimal Transport},
  journal = {Discrete \& Continuous Dynamical Systems - A},
  year    = {2014},
  volume  = {34},
  number  = {4},
  pages   = {1533--1574},
  doi     = {10.3934/dcds.2014.34.1533}
}

@article{PeyreCuturi2019,
  author  = {Peyr{\'e}, Gabriel and Cuturi, Marco},
  title   = {Computational Optimal Transport},
  journal = {Foundations and Trends in Machine Learning},
  year    = {2019},
  volume  = {11},
  number  = {5-6},
  pages   = {355--607},
  doi     = {10.1561/2200000073}
}

@book{deHaanFerreira2006,
  author    = {de Haan, Laurens and Ferreira, Ana},
  title     = {Extreme Value Theory: An Introduction},
  publisher = {Springer},
  address   = {New York},
  year      = {2006}
}

@book{Resnick2008,
  author    = {Resnick, Sidney I.},
  title     = {Extreme Values, Regular Variation and Point Processes},
  publisher = {Springer},
  address   = {New York},
  year      = {2008},
  edition   = {2}
}

@book{Villani2009,
  author    = {Villani, C{\'e}dric},
  title     = {Optimal Transport: Old and New},
  publisher = {Springer},
  year      = {2009},
  series    = {Grundlehren der mathematischen Wissenschaften},
  volume    = {338}
}

@book{Galichon2016,
  author    = {Galichon, Alfred},
  title     = {Optimal Transport Methods in Economics},
  publisher = {Princeton University Press},
  address   = {Princeton, NJ},
  year      = {2016}
}

@article{Granovetter1973,
  author  = {Granovetter, Mark S.},
  title   = {The Strength of Weak Ties},
  journal = {American Journal of Sociology},
  year    = {1973},
  volume  = {78},
  number  = {6},
  pages   = {1360--1380}
}

@article{Montgomery1991,
  author  = {Montgomery, James D.},
  title   = {Social Networks and Labor-Market Outcomes: Toward an Economic Analysis},
  journal = {American Economic Review},
  year    = {1991},
  volume  = {81},
  number  = {5},
  pages   = {1408--1418}
}

@article{Topa2001,
  author  = {Topa, Giorgio},
  title   = {Social Interactions, Local Spillovers and Unemployment},
  journal = {Review of Economic Studies},
  year    = {2001},
  volume  = {68},
  number  = {2},
  pages   = {261--295}
}

@article{CalvoArmengolJackson2004,
  author  = {Calv{\'o}-Armengol, Antoni and Jackson, Matthew O.},
  title   = {The Effects of Social Networks on Employment and Inequality},
  journal = {American Economic Review},
  year    = {2004},
  volume  = {94},
  number  = {3},
  pages   = {426--454}
}

@article{IoannidesLoury2004,
  author  = {Ioannides, Yannis M. and Loury, Linda D.},
  title   = {Job Information Networks, Neighborhood Effects, and Inequality},
  journal = {Journal of Economic Literature},
  year    = {2004},
  volume  = {42},
  number  = {4},
  pages   = {1056--1093}
}

@article{BuhaiVanderLeij2023,
  author  = {Buhai, I. Sebastian and van der Leij, Marco J.},
  title   = {A Social Network Analysis of Occupational Segregation},
  journal = {Journal of Economic Dynamics and Control},
  year    = {2023},
  volume  = {147},
  pages   = {104593},
  doi     = {10.1016/j.jedc.2022.104593}
}

@article{BolteImmorlicaJackson2020,
  author  = {Bolte, Lukas and Immorlica, Nicole and Jackson, Matthew O.},
  title   = {The Role of Referrals in Immobility, Inequality, and Inefficiency in Labor Markets},
  journal = {Journal of Labor Economics},
  year    = {2024},
  note    = {Advance online publication / forthcoming in Journal of Labor Economics},
  doi     = {10.1086/733048}
}

@article{RothschildStiglitz1970,
  author  = {Rothschild, Michael and Stiglitz, Joseph E.},
  title   = {Increasing Risk: I. A Definition},
  journal = {Journal of Economic Theory},
  year    = {1970},
  volume  = {2},
  number  = {3},
  pages   = {225--243},
  doi     = {10.1016/0022-0531(70)90038-4}
}

@book{ShakedShanthikumar2007,
  author    = {Shaked, Moshe and Shanthikumar, J. George},
  title     = {Stochastic Orders},
  publisher = {Springer},
  address   = {New York},
  year      = {2007},
  series    = {Springer Series in Statistics},
  doi       = {10.1007/978-0-387-34675-5}
}

@article{Pickands1975,
  author  = {Pickands, James},
  title   = {Statistical Inference Using Extreme Order Statistics},
  journal = {The Annals of Statistics},
  year    = {1975},
  volume  = {3},
  number  = {1},
  pages   = {119--131},
  doi     = {10.1214/aos/1176343003}
}

@article{BalkemaDeHaan1974,
  author  = {Balkema, August A. and de Haan, Laurens},
  title   = {Residual Life Time at Great Age},
  journal = {The Annals of Probability},
  year    = {1974},
  volume  = {2},
  number  = {5},
  pages   = {792--804},
  doi     = {10.1214/aop/1176996548}
}

@article{DonskerVaradhan1975,
  author  = {Donsker, Monroe D. and Varadhan, S. R. S.},
  title   = {Asymptotic Evaluation of Certain Markov Process Expectations for Large Time, I},
  journal = {Communications on Pure and Applied Mathematics},
  year    = {1975},
  volume  = {28},
  number  = {1},
  pages   = {1--47},
  doi     = {10.1002/cpa.3160280102}
}

@book{ArnoldBalakrishnanNagaraja1998,
  author    = {Arnold, Barry C. and Balakrishnan, N. and Nagaraja, H. N.},
  title     = {Records},
  publisher = {Wiley},
  address   = {New York},
  year      = {1998}
}

@article{BarndorffNielsen1964,
  author  = {Barndorff-Nielsen, O.},
  title   = {On the Limit Distribution of the Maximum of a Random Number of Independent Random Variables},
  journal = {Acta Mathematica Academiae Scientiarum Hungaricae},
  year    = {1964},
  volume  = {15},
  pages   = {399--403},
  doi     = {10.1007/BF01897148}
}

@article{Galambos1973,
  author  = {Galambos, J{\'a}nos},
  title   = {The Distribution of the Maximum of a Random Number of Random Variables with Applications},
  journal = {Journal of Applied Probability},
  year    = {1973},
  volume  = {10},
  number  = {1},
  pages   = {122--129},
  doi     = {10.2307/3212500}
}

@article{SilvestrovTeugels1998,
  author  = {Silvestrov, Dmitrii and Teugels, Jozef L.},
  title   = {Limit Theorems for Extremes with Random Sample Size},
  journal = {Advances in Applied Probability},
  year    = {1998},
  volume  = {30},
  number  = {3},
  pages   = {777--806},
  doi     = {10.1239/aap/1035228129}
}

@article{RachevResnick1991,
  author  = {Rachev, Svetlozar T. and Resnick, Sidney I.},
  title   = {Max-Geometric Infinite Divisibility and Stability},
  journal = {Stochastic Models},
  year    = {1991},
  volume  = {7},
  number  = {2},
  pages   = {191--218},
  doi     = {10.1080/15326349108807184}
}

@article{Megyesi2002,
  author  = {Megyesi, Zolt{\'a}n},
  title   = {Domains of Geometric Partial Attraction of Max-Semistable Laws: Structure, Merge and Almost Sure Limit Theorems},
  journal = {Journal of Theoretical Probability},
  year    = {2002},
  volume  = {15},
  number  = {4},
  pages   = {973--1005},
  doi     = {10.1023/A:1020692805345}
}

@article{GhosalNutzBernton2022,
  author  = {Ghosal, Promit and Nutz, Marcel and Bernton, Espen},
  title   = {Stability of Entropic Optimal Transport and Schr{\"o}dinger Bridges},
  journal = {Journal of Functional Analysis},
  year    = {2022},
  volume  = {283},
  number  = {9},
  pages   = {109622},
  doi     = {10.1016/j.jfa.2022.109622}
}

@inproceedings{Cuturi2013,
  author    = {Cuturi, Marco},
  title     = {Sinkhorn Distances: Lightspeed Computation of Optimal Transport},
  booktitle = {Advances in Neural Information Processing Systems 26},
  year      = {2013},
  pages     = {2292--2300}
}

@article{EinmahlHe2023,
  author  = {Einmahl, John H. J. and He, Yi},
  title   = {Extreme Value Estimation for Heterogeneous Data},
  journal = {Journal of Business \& Economic Statistics},
  year    = {2023},
  volume  = {41},
  number  = {1},
  pages   = {255--269},
  doi     = {10.1080/07350015.2021.2008408}
}

@article{FournierGuillin2015,
  author  = {Fournier, Nicolas and Guillin, Arnaud},
  title   = {On the Rate of Convergence in Wasserstein Distance of the Empirical Measure},
  journal = {Probability Theory and Related Fields},
  year    = {2015},
  volume  = {162},
  number  = {3-4},
  pages   = {707--738},
  doi     = {10.1007/s00440-014-0583-7}
}

@misc{BeckerMangin2023,
  author       = {Becker, Louis and Mangin, Sephorah},
  title        = {The Effect of Search Frictions on Extreme Outcomes},
  year         = {2023},
  howpublished = {Working paper},
  note         = {Originally circulated in 2023}
}

@misc{MansanarezPolySwan2025,
  author       = {Mansanarez, Paul and Poly, Guillaume and Swan, Yvik},
  title        = {Stein's Method for Fr{\'e}chet Approximation: A Regularly Varying Functions Approach},
  year         = {2025},
  howpublished = {arXiv preprint arXiv:2510.14016},
  note         = {Preprint}
}

@book{Santambrogio2015,
  author    = {Santambrogio, Filippo},
  title     = {Optimal Transport for Applied Mathematicians: Calculus of Variations, PDEs, and Modeling},
  publisher = {Birkh{"a}user},
  address   = {Basel},
  year      = {2015},
  series    = {Progress in Nonlinear Differential Equations and Their Applications},
  volume    = {87},
  doi       = {10.1007/978-3-319-20828-2}
}
\endgroup

\end{document}